\def\isarxiv{1}
\ifdefined\isarxiv
\documentclass[11pt]{article}
\else
\documentclass{article}

     \PassOptionsToPackage{numbers, compress}{natbib}

\usepackage[nonatbib, final]{neurips_2022}
\fi

\usepackage[utf8]{inputenc} 
\usepackage[T1]{fontenc}    
\usepackage{amssymb,amsmath,amsthm,amsfonts}
\usepackage{bm}
\usepackage{textgreek}
\usepackage{mathtools}
\usepackage{enumitem}
\usepackage[font=small]{caption}
\usepackage[labelformat=simple]{subcaption}
\usepackage{float}
\usepackage[numbers,comma,sort&compress]{natbib}
\usepackage{algorithm}
\usepackage{algpseudocode}
\usepackage{lineno}

\usepackage{physics}
\usepackage{footnote}
\usepackage{xcolor}
\usepackage{mathrsfs}
\usepackage{bbm}
\usepackage[colorlinks]{hyperref}
\definecolor{mydarkblue}{rgb}{0,0.08,0.45}
\ifdefined\isarxiv
\hypersetup{colorlinks=true,citecolor=red,linkcolor=mydarkblue} 
\usepackage[margin=1in]{geometry}
\else
\hypersetup{colorlinks=true, citecolor=mydarkblue,linkcolor=mydarkblue}
\fi
\usepackage{cleveref}
\usepackage{thm-restate}

\newtheorem{fact}{Fact}[section]

\newtheorem{theorem}{Theorem}
\newtheorem{definition}{Definition}
\newtheorem{lemma}{Lemma}
\newtheorem{proposition}{Proposition}

\newtheorem{corollary}{Corollary}
\newtheorem{remark}{Remark}

\newcommand{\eq}[1]{(\ref{eq:#1})}

\newcommand{\wt}[1]{\widetilde{#1}}
\newcommand{\thm}[1]{\hyperref[thm:#1]{Theorem~\ref*{thm:#1}}}
\newcommand{\cor}[1]{\hyperref[cor:#1]{Corollary~\ref*{cor:#1}}}
\newcommand{\defn}[1]{\hyperref[defn:#1]{Definition~\ref*{defn:#1}}}
\newcommand{\lem}[1]{\hyperref[lem:#1]{Lemma~\ref*{lem:#1}}}
\newcommand{\prop}[1]{\hyperref[prop:#1]{Proposition~\ref*{prop:#1}}}
\newcommand{\assum}[1]{\hyperref[assum:#1]{Assumption~\ref*{assum:#1}}}
\newcommand{\fig}[1]{\hyperref[fig:#1]{Figure~\ref*{fig:#1}}}
\newcommand{\tab}[1]{\hyperref[tab:#1]{Table~\ref*{tab:#1}}}
\newcommand{\alg}[1]{\hyperref[alg:#1]{Algorithm~\ref*{alg:#1}}}
\renewcommand{\sec}[1]{\hyperref[sec:#1]{Section~\ref*{sec:#1}}}
\newcommand{\append}[1]{\hyperref[append:#1]{Appendix~\ref*{append:#1}}}
\newcommand{\fac}[1]{\hyperref[fac:#1]{Fact~\ref*{fac:#1}}}
\newcommand{\lin}[1]{\hyperref[lin:#1]{Line~\ref*{lin:#1}}}

\def\>{\rangle}
\def\<{\langle}
\def\trans{^{\top}}

\newcommand{\vect}[1]{\ensuremath{\mathbf{#1}}}
\newcommand{\x}{\ensuremath{\mathbf{x}}}
\newcommand{\eps}{\epsilon}

\newcommand{\N}{\mathbb{N}}

\newcommand{\R}{\mathbb{R}}
\newcommand{\C}{\mathbb{C}}

\newcommand{\E}{\mathbb{E}}

\newcommand{\K}{\mathcal{K}}

\newcommand{\reg}{\mathcal{R}}

\DeclareMathOperator{\poly}{poly}

\renewcommand{\d}{\mathrm{d}}

\ifdefined\isarxiv
\title{Quantum Speedups of Optimizing Approximately Convex Functions with Applications to Logarithmic Regret Stochastic Convex Bandits\thanks{A Preliminary version of this paper was accepted by the 36th Conference on Neural Information Processing Systems (NeurIPS 2022).}}
\else
\title{Quantum Speedups of Optimizing Approximately Convex Functions with Applications to Logarithmic Regret Stochastic Convex Bandits}
\fi

\ifdefined\isarxiv
\date{}
\author{
  Tongyang Li\thanks{\texttt{tongyangli@pku.edu.cn}. Peking University.}
  \and
  Ruizhe Zhang\thanks{\texttt{ruizhe@utexas.edu}. The University of Texas at Austin.}
}
\else
\author{%
Tongyang Li\\
Advanced Institute of Information Technology,\\
Center on Frontiers of Computing Studies,\\
and School of Computer Science,\\
Peking University\\
\texttt{tongyangli@pku.edu.cn}\\
\And
Ruizhe Zhang\\
Department of Computer Science\\
The University of Texas at Austin\\
\texttt{ruizhe@utexas.edu}
}
\fi
\begin{document}

\maketitle

\setcounter{footnote}{0}

\begin{abstract}
We initiate the study of quantum algorithms for optimizing approximately convex functions. Given a convex set $\K\subseteq\R^{n}$ and a function $F\colon\R^{n}\to\R$ such that there exists a convex function $f\colon\mathcal{K}\to\R$ satisfying $\sup_{x\in\K}|F(x)-f(x)|\leq \epsilon/n$, our quantum algorithm finds an $x^{*}\in\K$ such that
$F(x^{*})-\min_{x\in\K} F(x)\leq\epsilon$ using $\tilde{O}(n^{3})$ quantum evaluation queries to $F$. This achieves a polynomial quantum speedup compared to the best-known classical algorithms. As an application, we give a quantum algorithm for zeroth-order stochastic convex bandits with $\tilde{O}(n^{5}\log^{2} T)$ regret, an exponential speedup in $T$ compared to the classical $\Omega(\sqrt{T})$ lower bound. Technically, we achieve quantum speedup in $n$ by exploiting a quantum framework of simulated annealing and adopting a quantum version of the hit-and-run walk. Our speedup in $T$ for zeroth-order stochastic convex bandits is due to a quadratic quantum speedup in multiplicative error of mean estimation.
\end{abstract}


\section{Introduction}
Optimization theory is a central research topic in computer science, mathematics, operations research, etc. Currently, many efficient algorithms for optimizing convex functions have been proposed (see for instance~\cite{Boyd2004}), but much less is known for nonconvex optimization. In this paper, we investigate polynomial-time algorithms for optimizing approximately convex functions. On the one hand, such algorithms enjoy robustness and cover many natural scenarios including stochastic convex optimization, empirical risk minimization, etc. On the other hand, approximately convex optimization paves the way of understanding nonconvex optimization in the general case.

Specifically, let $\K\subseteq\R^{n}$ be a convex set. We call $F\colon\R^{n}\to\R$ as an \emph{approximately convex function} over $\K$ if there is a convex function $f\colon\K\to\R$ such that
\begin{align}\label{eq:def_approx_convex}
\sup_{x\in\K}|F(x)-f(x)|\leq \epsilon/n.
\end{align}
Throughout the paper, we assume that $f$ is $L$-Lipschitz with respect to $\ell_{\infty}$ norm, i.e., $|f(x)-f(y)|\leq L\|x-y\|_{\infty}$ for any $x,y\in\K$. We assume that $\mathcal{B}_{2}(0,1)\subseteq\K\subseteq\mathcal{B}_{2}(0,R)$, i.e., the convex body $\K$ contains the unit ball centered at 0 and is contained by a ball of radius $R$ centered at 0. Unless otherwise mentioned, we aim at algorithms with $\poly(\log R)$ dependence in $R$.

It is standard to assume the \emph{zeroth-order oracle} of $F$, which returns the function value $F(x)$ given an input $x$. As far as we know, the state-of-the-art algorithm for finding an $x^{*}\in\K$ such that
$F(x^{*})-\min_{x\in\K} F(x)\leq\epsilon$
with high probability was proposed by Belloni et al.~\cite{belloni2015escaping}, which takes $\tilde{O}(n^{4.5})$\footnote{The $\tilde{O}$ and $\tilde{\Omega}$ notation omits poly-logarithmic terms, i.e., $\tilde{O}(g)=O(g\poly(\log g))$ and $\tilde{\Omega}(g)=\Omega(g\poly(\log g))$. Unless otherwise mentioned, both notations also omit $\poly(\log 1/\epsilon)$ terms.} queries to the zeroth-order oracle. On the other hand, Ref.~\cite{risteski2016algorithms} proved that if the approximation error in \eq{def_approx_convex} is at least $\tilde{\Omega}(\max\{\epsilon/n,\epsilon^{2}/\sqrt{n}\})$, there exists a function $F$ which no algorithm can find a point $x^{*}\in\K$ such that $F(x^{*})-\min_{x\in\K} F(x)\leq\epsilon$ using $\poly(n,1/\epsilon)$ queries to $F$. In other words, the $\epsilon/n$ term in \eq{def_approx_convex} is fundamental, and it has been the standard assumption of studying approximately convex optimization in~\cite{belloni2015escaping,risteski2016algorithms}.

Based on approximate convex functions, a closely related scenario is stochastic convex functions, where we have a function $F\colon\R^{n}\to\R$ such that $F(x)=f(x)+\eps_{x}$. Here $f\colon\R^{n}\to\R$ is a convex function and $\eps_{x}$ is a sub-Gaussian random variable with parameter $\sigma$, i.e., $\E[\exp(\lambda\eps_{x})]\leq\exp(\sigma^{2}\lambda^{2}/2)$ for any $\lambda$. This implies that
\begin{align}\label{eq:subGaussian}
\Pr[|\eps_{x}|\geq \sigma t]\leq 2\exp(-t^{2}/2)\quad\forall t\geq 0.
\end{align}
Using this fact, Belloni et al.~\cite{belloni2015escaping} gave an algorithm for stochastic convex optimization with $\tilde{O}(n^{7.5}/\epsilon^{2})$ queries to a zeroth-order oracle of $F$.

\paragraph{Contributions.}
In this paper, we conduct a systemic study of quantum algorithms for optimization of approximately convex functions, with applications to zeroth-order stochastic convex bandits. Quantum computing is a rapidly advancing technology, the capability of quantum computers is dramatically increasing and recently reached ``quantum supremacy" by Google~\cite{arute2019supremacy} and USTC~\cite{zhong2020quantum}. In optimization theory, quantum advantages have been proven for semidefinite programs~\cite{vanApeldoorn2018SDP,vanApeldoorn2017quantum,brandao2017SDP,brandao2016quantum}, general convex optimization~\cite{vanApeldoorn2018convex,chakrabarti2018quantum}, the escaping from saddle point problem in optimization~\cite{zll21}, etc. Nevertheless, as far as we know, quantum algorithms for approximately convex optimization and stochastic convex optimization are widely open. In this paper, we consider these problems using the quantum zeroth-order evaluation oracle $O_{F}$, a standard model used in previous quantum computing literature~\cite{vanApeldoorn2018convex,cch19,chakrabarti2018quantum,zll21}:
\begin{align}\label{eq:quantum-eval}
O_{F}|x,y\>=|x,F(x)+y\>\quad\ \forall x\in\R^{n}, y\in\R.
\end{align}
Here $|\cdot\>$ is the Dirac notation, and preliminaries of quantum computing will be covered in~\sec{prelim}. Intuitively, Eq.~\eq{quantum-eval} can take inputs with form $\sum_{i=1}^{m}\vect{c}_{i}\ket{\x_{i}}\otimes \ket{0}$ where $\x_{i}\in\R^{n}\ \forall i\in[m]$ and $\sum_{i=1}^{m}|\vect{c}_{i}|^{2}=1$, and if we measure the outcome quantum state, we get $F(\vect{x}_{i})$ with probability $|\vect{c}_{i}|^{2}$. In particular, the quantum zeroth-order oracle allows the ability to query different locations in \emph{superposition}, which is stronger than the classical counterpart (i.e., $m=1$). Nevertheless, if the classical zeroth-order oracle can be implemented by explicit arithmetic circuits, the quantum oracle in \eq{quantum-eval} can be implemented by quantum circuits of the same size up to logarithmic factors. As for stochastic convex functions, similar to \eq{quantum-eval}, we assume the following oracle:
\begin{align}\label{eq:quantum-eval-stoc}
O_{f}|x,y\>=|x\>\int_{\xi\in\R}\sqrt{g_{x}(\xi)}|f(x)+y+\xi\>\,\d\xi\quad\ \forall x\in\R^{n}, y\in\R,
\end{align}
where for any $x\in\R^{n}$, $g_{x}(\xi)$ follows a sub-Gaussian distribution as in \eq{subGaussian}. 

Our first result is a quantum algorithm for optimizing approximately convex functions:
\begin{theorem}\label{thm:approx-convex}
With probability at least $0.9$, we can find an $x^{*}\in\K$ such that
\begin{align}
F(x^{*})-\min_{x\in\K} F(x)\leq\epsilon
\end{align}
using $\tilde{O}(n^{3})$ queries to the quantum evaluation oracle \eq{quantum-eval}. 
\end{theorem}
We remark that the succss probability can be easily boosted up to $1-\delta$ for any $\delta\in (0,1)$, by paying an extra $\log(1/\delta)$ factor in the quantum query complexity. 
Compared to the best-known classical result by Belloni et al.~\cite{belloni2015escaping} with query complexity $\tilde{O}(n^{4.5})$, we achieve a polynomial quantum speedup in terms of $n$. Technically, Belloni et al.'s algorithm is based on the simulated annealing process and using the Hit-and-Run walk to generate samples in each stage. In this work, we give a \emph{user-friendly} version of the quantum walk framework (\thm{qwalk_intro}) that can be applied very easily to obtain quantum speedup for the mixing of classical Markov chains. We then analyze the warmness and the overlap between adjacent Hit-and-Run walks in the simulated annealing process, showing that our quantum walk framework is applicable to each classical Hit-and-Run sampler. By implementing the random walk quantumly, we improve the query complexity of the sampling procedure from $\wt{O}(n^3)$ to $\wt{O}(n^{1.5})$ (\thm{q_hit_and_run_intro}). We also design a non-destructive rounding procedure (\lem{q_round_intro}) in each stage of the simulated annealing that simulates the classical rounding procedure in \cite{belloni2015escaping} but will not destroy the quantum states of the quantum walks. Combining them together gives the $\wt{O}(n^3)$-query quantum algorithm for minimizing an approximately convex function. Our result can also be applied to give an $\wt{O}(n^3)$-query quantum algorithm for optimizing approximately convex functions with decreasing fluctuations. 

See \sec{approx-convex-quantum} for more details and the proof of \thm{approx-convex}.

Furthermore, to estimate the mean of a random variable up to multiplicative error $\epsilon$, classical algorithms need to take $1/\epsilon^{2}$ samples due to concentration inequalities such as Chernoff's bound, whereas quantum algorithms can take roughly $1/\epsilon$ queries (see \prop{subGaussian}). As a result, we obtain polynomial quantum speedup in both $n$ and $1/\epsilon$ for stochastic convex optimization:
\begin{corollary}\label{cor:approx-convex-stoc}
With probability at least $0.8$, we can find an $x^{*}\in\K$ such that
\begin{align}
f(x^{*})-\min_{x\in\K} f(x)\leq\epsilon
\end{align}
using $\tilde{O}(n^{5}/\epsilon)$ queries to the quantum stochastic evaluation oracle \eq{quantum-eval-stoc}.
\end{corollary}

We apply \cor{approx-convex-stoc} to solve the zeroth-order stochastic convex bandit problem, which is a widely studied bandit model (see e.g., \cite{hazan2016optimal,lattimore2020improved,lattimore2021improved}). The problem is defined as follows. Let $\K\subseteq\R^{n}$ be a convex body and $f\colon\K\to[0,1]$ be a convex function. Here $\mathcal{B}_{2}(0,1)\subseteq\K\subseteq\mathcal{B}_{2}(0,R)$. An online learner and environment interact alternatively over $T$ rounds. In each round $t\in[T]$, the learner makes a query to the quantum stochastic evaluation oracle \eq{quantum-eval-stoc}, and returns a value $x_{t}\in\K$ as the current guess. The learner aims to minimize the regret
\begin{align}\label{eq:reg-T}
\reg_{T}:=\E\left[\sum_{t=1}^{T}(f(x_{t})-f^{*})\right],\quad\text{where } f^{*}=\min_{x\in\K}f(x),
\end{align}
and the expectation is taken over all randomness. The classical state-of-the-art algorithm using a classical stochastic evaluation in each round achieves a regret bound of $\tilde{O}(n^{4.5}\sqrt{T})$~\cite{lattimore2021improved}, and there is a classical lower bound $\Omega(n\sqrt{T})$ on the regret~\cite{dani2009lower}. Here we prove:
\begin{theorem}\label{thm:stochastic-bandit}
There is a quantum algorithm for which $\reg_{T}=\tilde{O}(n^{5}\log(T)\log(TR))$.
\end{theorem}
This achieves $\poly(\log T)$ regret for zeroth-order stochastic convex bandits, \textbf{an exponential quantum advantage} in terms of $T$ compared to classical zeroth-order stochastic convex bandits. 
As far as we know, we give the first quantum algorithms with poly-logarithmic regret bound on online learning problems. An independent work~\cite{quantummab2022} gave quantum algorithms with poly-logarithmic regret for multi-armed bandits and stochastic linear bandits, but these two types of bandits concern reward of discrete objects and linear functions, respectively, which are fundamentally different from the stochastic convex bandits we study.

To achieve this $\poly(\log T)$ regret for zeroth-order stochastic convex bandits, we divide the $T$ iterations into $\lfloor\log_{2}T\rfloor$ intervals with doubling length $1,2,4,\ldots$. We use the quantum queries from a previous interval to run our quantum stochastic convex optimization algorithm and use the output as the guess for all iterations in the next interval. Since we can achieve linear dependence in $1/\epsilon$ in \cor{approx-convex-stoc}, it can be calculated that the total regret in each iteration is at most $\tilde{O}(n^{5}\log(TR))$, leading to our claim. The proof details are given in \sec{stoc-bandit-regret}.

\paragraph{Open questions.}
Our work raises several natural questions for future investigation:
\begin{itemize}
\item Can we further improve the dimension dependence of our approximate convex optimization algorithm? The current quantum speedup mainly leverages the quantum hit-and-run walk; it is of general interest to understand whether we can gain quantum advantage at other steps, or whether the convergence analysis of the hit-and-run walk per se can be improved.
\item Can we improve the dimension dependence of the regret of our zeroth-order stochastic convex bandits? It is natural to check whether our $n^5$ term can be improved by quantizing the classical state-of-the-art~\cite{lattimore2021improved} and other recent works.
\item Can we give fast quantum algorithms for more general nonconvex optimization problems, or with poly-logarithmic regret for more general bandit problems?
\end{itemize}

\ifdefined\isarxiv
\else
\paragraph{Limitations and societal impacts.}
This work is purely theoretical. Researchers working on theoretical aspects of quantum computing and optimization theory may benefit from our results at the moment. In the long term, once we have fault-tolerant quantum computers, our results can be applied to approximate convex optimization and stochastic bandit scenarios arising in the real world. As far as we are aware, our work does not have negative societal impacts.
\fi


\section{Preliminaries}\label{sec:prelim}
\subsection{Basics of quantum computing}
We briefly introduce basic notations and concepts of quantum computing in this section. More details are covered in standard textbooks, for instance~\cite{nielsen2000book}.

Quantum computing can be formulated in terms of linear algebra. Specifically, the computational basis of space $\C^{d}$ can be denoted by $\{\vec{e}_{1},\ldots,\vec{e}_{d}\}$, where $\vec{e}_{i}=(0,\ldots,1,\ldots,0)\trans$ with the $i^{\text{th}}$ entry being 1 and other entries being 0. These basis vectors are typically written by the \emph{Dirac notation}, where we denote $\vec{e}_{i}$ by $|i\>$ (called a ``ket"), and denote $\vec{e}_{i}\trans$ by $\<i|$ (called a ``bra").

A $d$-dimensional \emph{quantum state} is an $\ell_{2}$-norm unit vector in $\C^{d}$, e.g., $|v\>=(v_{1},\ldots,v_{d})\trans$ such that $\sum_{i}|v_{i}|^{2}=1$. For each $i$, $v_{i}$ is called the \emph{amplitude} in $|i\>$. The \emph{tensor product} of quantum states is their Kronecker product: if $|u\>\in\C^{d_{1}}$ and $|v\>\in\C^{d_{2}}$, then
\begin{align*}
|u\>\otimes|v\>:=(u_{1}v_{1},u_{1}v_{2},\ldots,u_{d_{1}}v_{d_{2}})\trans\in\C^{d_{1}}\otimes\C^{d_{2}}.
\end{align*}
We often omit the operator $\otimes$ and simply write $\ket{u}\ket{v}$ or $\ket{u,v}$ for being concise.

In general, the definition of quantum states can be extended to a continuous domain. For instance,
\begin{align*}
|v\>=\int_{\R^{n}}v_{x}|x\>\d x
\end{align*}
represents a quantum state as long as $\int_{\R^{n}}|v_{x}|^{2}\d x=1$. To keep quantum states normalized in $\ell_{2}$ norm, operations in quantum computing are \emph{unitary transformations}.

\subsection{Classical and quantum walks}
A classical Markov over the space $\Omega$ is a sequence of random variables $\{X_i\}_{i\in \mathbb{N}}$ such that for any $i>0$ and $x_0,\dots,x_i\in \Omega$,
\begin{align*}
  \Pr[X_i=x_i~|~X_0=x_0,\cdots,X_{i-1}=x_{i-1}] = \Pr[X_i=x_i~|~X_{i-1}=x_{i-1}].
\end{align*}
The Markov chain can be represented by its stochastic transition matrix $P$ such that $\sum_{y\in \Omega} P(x,y)=1$ for any $x\in \Omega$. A distribution $\pi$ is \emph{stationary} if it satisfies $\sum_{x\in \Omega} \pi(x)P(x,y) = \pi(y)$ for any $y\in \Omega$. A Markov chain is \emph{reversible} if its stationary distribution satisfies the following detailed balance condition:
\begin{align*}
  \pi(x)P(x,y) = \pi(y)P(y,x)~~~\forall x,y\in \Omega.
\end{align*}
The \emph{mixing time} of a Markov chain with initial distribution $\pi_0$ is the number of steps $t=t(\epsilon)\in \mathbb{N}$ such that the total variation distance between the time-$t$ distribution and the stationary distribution is at most $\epsilon$ for any $\epsilon\in (0,1)$, i.e.,
\begin{align*}
  d_\mathrm{TV}(P^{t}\pi_0, \pi)\leq \epsilon.
\end{align*}

In quantum, we can also define the discrete-time quantum walk \cite{sze04}, which is a quantum-analogue of classical Markov chain. More precisely, we use a quantum state to represent a classical probability distribution:
\begin{align*}
  \{\pi(x)\}_{x\in \Omega}~\longleftrightarrow~\ket{\pi}=\sum_{x\in \Omega}\sqrt{\pi(x)}\ket{x}.
\end{align*}
In quantum walk, there are two operations:
\begin{itemize}
  \item {\bf Reflect:} An operator $R$ that reflects the quantum states with respect to the subspace $\mathrm{span} \{\ket{x}\ket{\psi_x}\}_{x\in \Omega}$, where $\ket{\psi_x}=\sum_{y\in \Omega}\sqrt{P(x,y)}\ket{y}$.
  \item {\bf Swap:} An operator $S$ that swap the two quantum registers: $S\ket{x}\ket{y}=\ket{y}\ket{x}$.
\end{itemize}
Then, the quantum walk operator $W$ is defined as: $W:=S\circ R$. Intuitively, the first quantum register contains the current position of the random walk, and the second register contains the previous position. In each step of quantum walk, the $R$ operator makes a superposition in the second register of the next step positions with amplitudes proportional to their transition probabilities. Then, the $S$ operator swaps the two quantum registers, using the first register to store the new positions and the second register to store the old position. In this way, we complete one-step of the random walk coherently (in superposition).

The quantum advantage of the quantum walk comes from the spectrum of $W$. Note that the state of stationary distribution $\sum_{x\in \Omega}\sqrt{\pi(x)}\ket{x}\ket{\psi_x}$ is invariant under $W$. In other words, it is an eigenvector with eigenvalue $1$ (eigenphase 0). On the other hand, the other eigenvectors of $W$ has eigenphase at least $\sqrt{2\delta}$, where $\delta$ is the spectral gap\footnote{The difference between the first and the second largest eigenvalues.} of the transition matrix $P$.  Therefore, applying the \emph{quantum phase estimation} algorithm using $O(1/\sqrt{\delta})$ calls to $W$ can distinguish the state corresponding to the stationary distribution and other eigenstates. See \cite{sze04,wocjan2008speedup,chi21} for more details.

\subsection{Hit-and-Run walk}
Hit-and-Run walk was introduced by R.L. Smith \cite{smi84}, and has a long line of reserach \cite{bhc87,bel93,che93,zsr93,Lovasz99,lv03,lv06,lv07,ayp17} for its mixing time and applications in sampling, optimization, and volume estimation. Intuitively, the Hit-and-Run walk is defined as follows:
\begin{enumerate}
  \item Pick a uniformly distributed random line $\ell$ through the current point $x$.
  \item Move to a random point $y$ along the line chosen from the restricted distribution $\pi_\ell$.
\end{enumerate}
More specifically, let $f:\R^n\rightarrow \R$ be a logconcave distribution density. Then, the distribution induced by restricting $f$ to the line $\ell$ is defined as follows:
\begin{align*}
  \pi_\ell(S):=\frac{\int_{S} f(x)\d x}{\int_\ell f(x)\d x}~~~\forall S\subset \ell.
\end{align*}
The following lemma show the transition probability of the Hit-and-Run walk.
\begin{lemma}[\cite{Lovasz99}]
  Let $f$ be the density of a logconcave distribution. If the current point of Hit-and-Run is $u$, then the density function of the distribution of the next point $x$ is
  \begin{align*}
    f_u(x)=\frac{2}{n\pi_n}\frac{f(x)}{\mu_f(u,x)\|x-u\|^{n-1}},
  \end{align*}
  where $\pi_n=\frac{\pi^{n/2}}{\Gamma(1+n/2)}$ and $\mu_f(u,x)$ is the $f$-measure of the chord through $u$ and $x$.
\end{lemma}


\section{Quantum Algorithm for Optimizing Approximately Convex Functions}\label{sec:approx-convex-quantum}
We give a polynomial quantum speedup for minimizing an approximately convex functions using the quantum walk algorithm \cite{sze04,wocjan2008speedup}. More specifically, we consider the quantum walk in continuous space.\footnote{It can be naturally discretized as we do in simulating a Markov chain on classical digital computers.} Let $P$ denote the (column) stochastic transition density of a reversible Markov chain, i.e.,
\begin{align*}
  \int_{\R^n} P(x,y)\d y = 1 ~~~\forall x\in \R^n,
\end{align*}
and let $\pi(x)$ denote the density of the stationary distribution. Then, we can implement a quantum walk unitary $W(P)$ such that its unique eigenvector with eigenvalue 1 (or equivalently eigenphase 0) is:
\begin{align*}
  \int_{\R^n} \sqrt{\pi(x)}\ket{x}\otimes \ket{\psi_x}\d x,
\end{align*}
where $\ket{\psi_x}:=\int_{\R^n}\sqrt{P(x,y)}\ket{y}\d y$ mixes all the points that can be moved from $x$, with amplitudes proportional to the transition probabilities. This quantum state can be considered as a \emph{coherent encoding} of the classical distribution $\pi$. The advantage of quantum walk comes from the fact that $W(P)$ has phase gap $\sqrt{\delta}$, where $\delta$ is the spectral gap of $P$. Therefore, by the quantum phase estimation algorithm \cite{kit95}, quantum walk can achieve quadratic speedup in $\delta$. In general, quantum walk algorithm can quadratically speedup the hitting time of classical Markov chain. For the mixing time, it requires some additional complicated constraints on the Markov chain and distributions (see e.g., \cite{aak01,wocjan2009quantum,obd18,clw19}). Based on previous studies on quantum walk mixing, we propose the following user-friendly quantum walk framework:

\begin{theorem}[User-friendly quantum walk framework]\label{thm:qwalk_intro}
  Let $M_0$ be the initial Markov chain with stationary distribution $\pi_0$, $M_1$ be the target Markov chain with stationary distribution $\pi_1$. Suppose $M_0,M_1$ satisfy the following properties:
  \begin{itemize}
    \item {\bf Mixing time: } $d_\mathrm{TV}(P_1^{t_0}\cdot \pi_0, \pi_1)\leq \epsilon$ and $d_\mathrm{TV}(P_0^{t_1}\cdot \pi_1, \pi_0)\leq \epsilon$.
    \item {\bf Warmness: } $\|\pi_0/\pi_1\|=O(1)$ and $\|\pi_1/\pi_0\|=O(1)$, where $\|\pi/\sigma\|:=\int_{\R^n} \frac{\pi(x)}{\sigma(x)}\pi(x)\d x$.
    \item {\bf Overlap: } $| \langle \pi_0|\pi_1\rangle|=\int_{\R^n}\sqrt{\pi_0(x)\pi_1(x)}\d x=\Omega(1)$.
  \end{itemize}
  Furthermore, suppose we have access to a unitary $U$ that prepares the initial state $\ket{\pi_0}=\int_{\R^n}\sqrt{\pi_0(x)}\ket{x}\d x$. Then, we can obtain a state $\ket{\wt{\pi}_1}$ with $\|\ket{\wt{\pi}_1}-\ket{\pi_1}\|_2\leq \epsilon$ using
  \begin{align*}
    O\left(\sqrt{t_0+t_1}\log^2(1/\epsilon)\right)
  \end{align*}
  calls to the quantum walk operators.
\end{theorem}

More details and proofs are deferred to \append{qwalk_toolbox}.

Next, we can use \thm{qwalk_intro} to speed-up the best-known classical algorithm for optimizing approximately convex functions~\cite{belloni2015escaping}, which has the following three levels:
\begin{itemize}
  \item {\bf High level: } Perform a simulated annealing with $K$ stages. At the $i$-th stage, the target distribution $\pi_{g_i}$ has density $\propto g_i(x)=e^{-F(x)/T_i}$, where $T_i:=(1-1/\sqrt{n})^i$.
  \item {\bf Middle level: } Use $N$ samples from $\pi_{g_i}$ to construct a linear transformation $\Sigma_i$, rounding the distribution to near-isotropic position.
  \item {\bf Low level: } Run the hit-and-run walk to evolve the distribution from $\pi_{g_{i-1}}$ to $\pi_{g_i}$.
\end{itemize}
Here, each step of the hit-and-run walk picks a uniformly random direction at the current point, and then walks on the 1-dimensional chord intersected by the direction and $\K$ with probability density proportional to the logconcave density.  We formally state the hit-and-run walk in \alg{hit-and-run}.

We focus on speeding-up the Low level using \thm{qwalk_intro}. Hence, we need to show that $\pi_{g_i}$ and $\pi_{g_{i+1}}$ satisfy the properties therein. First of all, it has been proved in \cite{belloni2015escaping} that $\|\pi_{g_i}/\pi_{g_{i+1}}\|=O(1)$ (\lem{warm_forward}).  We prove the following lemma with proof deferred to \append{qspeed_low_level}.
\begin{lemma}[Informal version of \lem{warm_backward}]
  Let $\pi_{g_i}$ be a distribution with density proportional to $g_i(x)=\exp(-F(x)/T_i)$, where $F(x)$ is $\beta$-approximately convex. Then, for any $0\leq i\leq K-1$,
    \begin{align*}
        \|\pi_{g_{i+1}}/\pi_{g_{i}}\|\leq 8\exp(2\beta/T_{i+1})\leq O(1).
    \end{align*}
\end{lemma}
Therefore, the warmness property is satisfied.

We also prove that the Markov chains in this annealing schedule are slowly evolving:
\begin{lemma}[Informal version of \lem{bound_overlap}]
  Let $\pi_{g_i}$ be a distribution with density proportional to $g_i(x)=\exp(-F(x)/T_i)$, where $F(x)$ is $\beta$-approximately convex.  Then, for any $0\leq i\leq K-1$,
    \begin{align*}
        |\langle \pi_i|\pi_{i+1}\rangle| \geq \exp(-(\beta/T_{i+1}+1)/2)=\Omega(1).
    \end{align*}
\end{lemma}
Hence, the overlap property is also satisfied.

With the warmness and classical analysis of the hit-and-run walk (\thm{hit_and_run_mixing_real}), we get that the classical mixing time from $\pi_{g_i}$ to $\pi_{g_{i+1}}$ and vice versa can be bounded by $\wt{O}(n^{3})$.

Moreover, we also show in \append{qspeed_low_level} and \lem{implement_q_hit_and_run} that each call to the quantum walk operator can be implemented by querying the evaluation oracle $O(1)$ times.

Thus, by \thm{qwalk_intro}, we get the following theorem:
\begin{theorem}[Low level quantum speedup, informal version of \thm{q_hit_and_run}]\label{thm:q_hit_and_run_intro}
  Let $\gamma\in (0, 1/e)$. Let $g_i(x)=\exp(-F(x)/T_i)$ be the density of $\pi_{g_i}$ with $F(x)$ being $\beta/2$-approximately convex. Let $T_i=(1-1/\sqrt{n})^i$ for $0\leq i\leq K$. Then, for each $0\leq i\leq K-1$, given a state $\ket{\pi_{g_i}}$, we can produce a state $|\hat{\sigma}_i^{(m)}\rangle$ such that
    \begin{align*}
        \||\pi_{g_{i+1}}\rangle-|\hat{\sigma}_i^{(m)}\rangle\|_2\leq O(\gamma),
    \end{align*}
    using $m=\wt{O}(n^{1.5})$ calls for the evaluation oracle of $F$.
\end{theorem}

In the Middle level, we need to use $N$ independent samples from $\pi_{g_i}$ to construct a linear transformation $\Sigma_i$ for rounding. However, we cannot directly measure the state $\ket{\pi_{g_i}}$, since it will destroy the quantum coherence. Instead, we propose a non-destructive approach to construct the linear transformation (the proof is deferred to \append{qspeed_mid_level}):
\begin{lemma}[Non-destructive rounding, informal version of \lem{qrounding}]\label{lem:q_round_intro}
  For each $i\in [K]$, the linear transformation $\Sigma_i$ can be obtained using $\wt{O}(N)$ copies of the states $\ket{\pi_{g_{i-1}}}$, with query complexity $\wt{O}(N\cdot n^{1.5})$. Moreover, the states $\ket{\pi_{g_{i-1}}}$ will be recovered with high probability.
\end{lemma}

Now, we can put all the components together and obtain a quantum algorithm for optimizing approximately convex function with $\wt{O}(n^{1.5})$ quantum query complexity (\alg{q_main_intro}). We sketch the proof in below and the formal proof is given in \append{q_opt_proof}:
\begin{proof}[Proof sketch of {\thm{approx-convex}}]
  Observe that at each annealing stage, the sample distribution is the same as the classical algorithm. Thus, by the classical analysis (\thm{classical_annealing_cost}) in \cite{belloni2015escaping}, the same optimization guarantee still holds for the quantum algorithm. Thus, if we take $K=\sqrt{n}\log(n/\epsilon)$ and $N=\wt{O}(n)$, the output $x_*$ of \textsc{QSimAnnealing} procedure satisfies:
  \begin{align*}
      F(x_*) - \min_{x}F(x)\leq O(\epsilon)
  \end{align*}
  with high probability.

  Then, consider the query complexity. There are $K$ stages in the annealing process, where each stage maintains $\wt{O}(N)$ samples (quantum states). By \thm{q_hit_and_run_intro}, evolving each state takes $\wt{O}(n^{1.5})$ queries. Therefore, the total query complexity is
  \begin{align*}
      K\cdot N\cdot \wt{O}(n^{1.5})=\wt{O}(n^{3}).
  \end{align*}

  Here, we assume that the convex body ${\cal K}$ is known, e.g., ${\cal K}=\R^n$ or $\mathbb{S}^n$. However, even if ${\cal K}$ is unknown, we can call its membership oralce to run our algorithm. More specifically, in constructing the initial state $\ket{\pi_0}$, we need to query the membership oracle for $\wt{O}(1)$ times. And since we prepare $N$ copies, this step takes $\wt{O}(N)$ queries in total. Then, in each step of the quantum walk, we need to query the membership oracle for $\wt{O}(1)$ times to determine the intersection point for the hit-and-run process. Thus, the number of queries to the membership oracle of ${\cal K}$ is the same as the the number of evaluation oracle queries. Hence, our algorithm will query the membership oracle for $\wt{O}(n^3)$ times in all. 

  As for the number of qubits to implement our algorithm, each state $\ket{\pi_0}$ uses $\wt{O}(n)$ qubits. Thus, $\wt{O}(n^2)$ qubits are used to store all states. And we need $O(n)$ ancilla qubits for the quantum walk unitaries. Therefore, our algorithm uses $\wt{O}(n^2)$ qubits in total.   
\end{proof}

\begin{algorithm}[htbp]
  \caption{Quantum speedup for approximately convex optimization (Informal version)}\label{alg:q_main_intro}
\begin{algorithmic}[1]
\Procedure{QSimAnnealing}{}\Comment{\thm{approx-convex}}
  \State $N\gets \wt{O}(n)$, $K\gets \sqrt{n}\log(n/\epsilon)$
  \State Prepare $N$ (approximately) copies of $\ket{\pi_0}$, denoted as $|\wt{\pi}_0^{(1)}\rangle,\dots,|\wt{\pi}_0^{(N)}\rangle$
  \For{$i\gets 1,\dots,K$}
      \State Use $\{|\wt{\pi}_0^{(j)}\rangle\}_{j\in[N]}$ to nondestructively construct $\Sigma_i$\Comment{\lem{q_round_intro}}
      \State Apply quantum walk to evolve the states $|\wt{\pi}_{i-1}^{(j)}\rangle$  to $|\wt{\pi}_{i}^{(j)}\rangle$\Comment{\thm{q_hit_and_run_intro}}
  \EndFor
  \State $x_K^j\gets$ measure the final state $|\wt{\pi}_K^{(j)}\rangle$ for $j\in [N]$
  \State \Return $\arg\min_{j\in [N]} F(x_K^j)$
\EndProcedure
\end{algorithmic}
\end{algorithm}

\paragraph{Optimization of approximately convex functions with decreasing fluctuations.}
Beyond the $\ell_{\infty}$-norm assumption for all $x\in\K$ in Eq.~\eq{def_approx_convex}, it is also possible to give efficient algorithms for optimizing other types of approximately convex functions. Specifically, Belloni et al.~\cite[Section 7]{belloni2015escaping} studied approximately convex functions with decreasing fluctuations.

Suppose that the function $f$ in~\eq{def_approx_convex} is 1-Lipschtiz and $\alpha$-strongly convex with minimum at $x_{\min}\in\K$:
\begin{align*}
f(x)-f(x_{\min})\geq\langle\nabla f(x_{\min}),x-x_{\min}\rangle+\tfrac{\alpha}{2}\|x-x_{\min}\|^{2}\geq\tfrac{\alpha}{2}\|x-x_{\min}\|^{2}.
\end{align*}
Define a measure of "non-convexity" of $F$ w.r.t to $f$ in an $n$-dimensional ball of radius $r$ near $x_{\min}$:
\begin{align*}
\Delta(r):=\sup_{x\in \mathcal{B}_{2}(x_{\min},r)}|F(x)-f(x)|.
\end{align*}
We can call \thm{approx-convex} iteratively. Suppose that at the start of the $t^{\text{th}}$ iteration we have a ball $\mathcal{B}_{2}(x_{t-1},2r_{t-1})$ satisfying $$\mathcal{B}_{2}(x_{\min},r_{t-1})\subset\mathcal{B}_{2}(x_{t-1},2r_{t-1})\subset\mathcal{B}_{2}(x_{\min},3r_{t-1}).$$ After executing \thm{approx-convex} in this iteration with $\wt{O}(n^{3})$ quantum queries, we reach a point $x_{t}$ such that with high probability $f(x_{t})-f(x_{\min})\leq Cn\Delta(3r_{t-1}))$ for some global constant $C>0$. Due to strong convexity, this gives a new radius $r_{t}$ recursively:
\begin{align*}
\frac{\alpha}{2Cn}r_{t}^{2}:=\frac{\alpha}{2Cn}\|x_{t}-x_{\min}\|^{2}\leq\Delta(3r_{t-1}).
\end{align*}
When $\Delta(r)=cr^{p}$ for some $c>0, p\in(0,2)$, the iteration stops when $r^{*}=(2\cdot 3^{p}cCn/\alpha)^{1/(2-p)}$, and when $\Delta(r)=c\log(1+dr)$ for some $c,d>0$, the iteration stops at $r^{*}$ satisfying $$(2cCn/\alpha)\log(1+3dr^{*})=(r^{*})^{2}.$$ The total number of quantum queries is still $\wt{O}(n^{3})$.


\section{Quantum Algorithm for Zeroth-Order Stochastic Convex Bandits}\label{sec:stoc-bandit-regret}
We first prove \cor{approx-convex-stoc} using quantum mean estimation with a Gaussian tail:
\begin{proposition}[Adapted from {\cite[Theorem 4.2]{hamoudi21}}]\label{prop:subGaussian}
Suppose that $X$ is a random variable on a probability space $(\Omega,p)$ with mean $\mu$ and variance $\sigma^{2}$. Suppose we have a unitary oracle $U$ satisfying $U|0\>=\int_{x\in\Omega}\sqrt{p(x)}|x\>\d x$. Then, for any $\Delta\in(0,1)$ and $\tau\in\N$ such that $\tau\geq\log(1/\Delta)$, there is a quantum algorithm that outputs a mean estimate $\tilde{\mu}$ such that
\begin{align*}
\Pr[|\tilde{\mu}-\mu|>\frac{\sigma\log(1/\Delta)}{\tau}]\leq\Delta,
\end{align*}
using $O(\tau\log^{3/2}(\tau)\log\log(\tau))$ queries to $U$.
\end{proposition}

\begin{proof}[Proof of {\cor{approx-convex-stoc}}]
We follow Section 6 of~\cite{belloni2015escaping} while use \thm{approx-convex} and \prop{subGaussian}. Specifically, for a parameter $0<\alpha<1$, we let $\mathcal{N}_{\alpha}$ be a box grid of $\K$ with side length $\alpha$. In other words, $\mathcal{N}_{\alpha}$ is $\alpha$-net of $\K$ in $\ell_{\infty}$ norm. Since $\K\subseteq\mathcal{B}_{2}(0,R)$, $|\mathcal{N}_{\alpha}|\leq(R/\alpha)^{n}$.

Note that for the sub-Gaussian random variable $\eps_{x}$ in \eq{subGaussian}, it has variance at most $4\sigma^{2}$ because
\begin{align*}
\E[|\eps_{x}|^{2}]=\int_{0}^{\infty}\Pr[|\eps_{x}|>\sqrt{s}]\d s\leq 2\int_{0}^{\infty} e^{-\frac{s}{2\sigma^{2}}}\d s=4\sigma^{2}.
\end{align*}
Upon a query $x'\in\K$, we define an oracle $O_{f}^{\tau,\alpha}$ which returns $f(x)+\tilde{\eps}_{x}$ for $x\in\mathcal{N}_{\alpha}$ which is closest to $x'$, and the $\tilde{\eps}_{x}$ here is obtained by applying \prop{subGaussian} with the unitary oracle $O_{f}$ in Eq.~\eq{quantum-eval-stoc} to estimate the function value $f(x)$. Specifically, with $\Delta=\exp(-t^{2})$ where $t$ is a parameter determined later, we have
\begin{align}
\Pr[|\tilde{\eps}_{x}|>\frac{\sigma t^{2}}{\tau}]\leq\exp(-t^{2}).
\end{align}
We note that in our algorithm based on the hit-and-run walk, with probability 1 we do not revisit the same point. As a result, $O_{f}^{\tau,\alpha}$ is no more powerful than $O_{f}$ since the learner only obtains information on $\mathcal{N}_{\alpha}$, and in the rest of the proof we assume $O_{f}^{\tau,\alpha}$ as the oracle we use. We take
\begin{align*}
\alpha=\epsilon/2nL,\quad t=\sqrt{n\ln(R/\alpha)+\ln 10}.
\end{align*}
Note that the value of $t$ promises that $\exp(-t^{2})(R/\alpha)^{n}\leq 0.1$. In other words, with probability at least $0.9$, we promise that
\begin{align}\label{eq:eps-max-bound}
\max_{x\in \mathcal{N}_{\alpha}}|\tilde{\eps}_{x}|\leq\frac{\sigma t^{2}}{\tau}=\frac{\sigma(n\ln(R/\alpha)+\ln 10)}{\tau}.
\end{align}
Finally, we take $\tau$ such that the RHS of \eq{eps-max-bound} equals to $\epsilon/2n$, which is equivalent to 
\begin{align*}
  \tau=\frac{2n\sigma(n\ln(R/\alpha)+\ln 10)}{\epsilon}=\tilde{O}(n^{2}/\epsilon).
\end{align*}
This will finally promise that
\begin{align}
\sup_{x\in\K}|F(x)-f(x)|\leq\max_{x\in\mathcal{N}_{\alpha}}|\tilde{\eps}_{x}|+\alpha L\leq\frac{\epsilon}{2n}+\frac{\epsilon}{2n}=\frac{\epsilon}{n},
\end{align}
meeting the condition of \thm{approx-convex}. Consequently, with probability at least $0.9\cdot 0.9>0.8$, we can find an $x^{*}\in\K$ such that $f(x^{*})-\min_{x\in\K} f(x)\leq\epsilon$ using $\tilde{O}(n^{3})\cdot\tau=\tilde{O}(n^{5}/\epsilon)$ queries to the quantum stochastic evaluation oracle \eq{quantum-eval-stoc}.
\end{proof}

\begin{algorithm}[t]
\caption{Quantum zeroth-order stochastic convex bandits}\label{alg:quantum-zeroth-order-bandit}
\begin{algorithmic}[1]
\Procedure{QBandits}{$T$}
  \State $m\leftarrow\lfloor\log_{2}T\rfloor$, $K\gets \lfloor \log_2(TR)\rfloor$
  \State $x_{1}\leftarrow 0$
\For{$i\gets 1,2,\dots, m+1$}
  \State $t\gets 2^{i-1}$
  \If{$t\leq m$}
    \State $L\gets 2^{i-1}$\Comment{$T_{i}:= \{2^{i-1},2^{i-1}+1,\ldots,2^{i}-1\}$}
  \Else
    \State $L\gets T-2^m+1$\Comment{$T_{m+1}:= \{2^{m},2^{m}+1,\ldots,T\}$}
  \EndIf
  \For{$j\gets 1,2,\ldots,K$}
  \State $y_j\gets \textsc{QMinStocConv}({\cal O}_{f}^{\tau, \alpha}, L/K)$\Comment{\cor{approx-convex-stoc} with $|T_i|/K$ queries}
	\For{$l\gets 0,1,\ldots, L/K$}
    \State $x_{t+l}\gets x_{2^{i-1}}$\Comment{Onliner learner's output at time $t+l$}
  \EndFor
  \If{$f(y_j)<f(x_{2^i})$}
    \State $x_{2^i}\gets y_j$
  \EndIf
  \EndFor
\EndFor
\EndProcedure
\end{algorithmic}
\end{algorithm}

{\noindent \emph{Proof of \thm{stochastic-bandit}.}}
We prove that \alg{quantum-zeroth-order-bandit} satisfies \thm{stochastic-bandit}.

Intuitively, we divide the $T$ rounds into $m+1$ intervals where $m\leftarrow\lfloor\log_{2}T\rfloor$, such that $[T]=\bigcup_{i=1}^{m+1}{\cal T}_{i}$ and ${\cal T}_{i}:=\{2^{i-1},2^{i-1}+1,\ldots,2^{i}-1\}$ for each $i\in[m]$. When executing in the interval ${\cal T}_{i}$, the output required by the online learner is always $x_{t}=x_{2^{i-1}}$, the $x$ at the end of the last interval. On the other hand, the queries in the current interval are applied to running the quantum stochastic convex optimization algorithm in \cor{approx-convex-stoc} and output a nearly-optimal solution with probability at least $1-O(1/T)$. With $|{\cal T}_{i}|=2^{i-1}$ queries at hand, we divide them into $\log (TR)$ repeats of \cor{approx-convex-stoc}, each using $2^{i-1}/\log (TR)$ queries in the quantum algorithm. As a result, each repeat $j\in[\log (TR)]$ outputs a value $\tilde{x}_{2^{i},j}$ such that
\begin{align*}
  f(\tilde{x}_{2^{i},j})-\min_{x\in\K} f(x)\leq\tilde{O}(n^{5}\log (TR)/2^{i-1})
\end{align*}
with probability at least 0.8. We take $x_{2^{i}}:=\arg\min_{j\in[\log T]}f(\tilde{x}_{2^{i},j})$. With probability at least $1-0.8^{\log (TR)}=1-O(1/TR)$, we have
\begin{align}
f(x_{2^{i}})-\min_{x\in\K} f(x)\leq\tilde{O}(n^{5}\log (TR)/2^{i}).
\end{align}
Going through all $i\in[m+1]$ intervals, by the union bound, with probability at least
\begin{align*}
  1-(m+1)\cdot O\Big(\frac{1}{TR}\Big)=1-O\Big(\frac{\log T}{TR}\Big),
\end{align*}
we have
\begin{align}
f(x_{2^{i}})-\min_{x\in\K} f(x)\leq\tilde{O}(n^{5}\log (TR)/2^{i-1})\quad\forall i\in[m+1].
\end{align}
In all, we get that the regret bound as desired:
\begin{align*}
\reg_{T}&=\E\left[\sum_{t=1}^{T}(f(x_{t})-f^{*})\right]\\
&\leq \left(1-O\left(\frac{\log T}{TR}\right)\right)\cdot\sum_{i=1}^{m+1}2^{i-1}\cdot\tilde{O}\left(\frac{n^{5}\log TR}{2^{i-1}}\right)+O\left(\frac{\log T}{TR}\right)\cdot T\cdot LR \\
&=\tilde{O}(n^{5}\log(T)\log(TR))+O(L\log T)\\
&=\tilde{O}(n^{5}\log(T)\log(TR)).
\tag*{$\square$}
\end{align*}


\section*{Acknowledgement}
We thank anonymous referees for valuable comments. TL was supported by a startup fund from Peking University. RZ was supported by the University Graduate Continuing Fellowship from UT Austin. 

\ifdefined\isarxiv
\bibliography{ApproxConvex}
\bibliographystyle{alpha}
\else
\small
\bibliography{ApproxConvex}
\bibliographystyle{myhamsplain}
\normalsize
\fi
\ifdefined\isarxiv
\else
\section*{Checklist}

\begin{enumerate}
\item For all authors...
\begin{enumerate}
  \item Do the main claims made in the abstract and introduction accurately reflect the paper's contributions and scope?
    \answerYes{}
  \item Did you describe the limitations of your work?
    \answerYes{}
  \item Did you discuss any potential negative societal impacts of your work?
    \answerYes{}
  \item Have you read the ethics review guidelines and ensured that your paper conforms to them?
    \answerYes{}
\end{enumerate}

\item If you are including theoretical results...
\begin{enumerate}
  \item Did you state the full set of assumptions of all theoretical results?
    \answerYes{}
        \item Did you include complete proofs of all theoretical results?
    \answerYes{}
\end{enumerate}

\item If you ran experiments...
\begin{enumerate}
  \item Did you include the code, data, and instructions needed to reproduce the main experimental results (either in the supplemental material or as a URL)?
    \answerNA{}
  \item Did you specify all the training details (e.g., data splits, hyperparameters, how they were chosen)?
    \answerNA{}
        \item Did you report error bars (e.g., with respect to the random seed after running experiments multiple times)?
    \answerNA{}
        \item Did you include the total amount of compute and the type of resources used (e.g., type of GPUs, internal cluster, or cloud provider)?
    \answerNA{}
\end{enumerate}

\item If you are using existing assets (e.g., code, data, models) or curating/releasing new assets...
\begin{enumerate}
  \item If your work uses existing assets, did you cite the creators?
   \answerNA{}
  \item Did you mention the license of the assets?
    \answerNA{}
  \item Did you include any new assets either in the supplemental material or as a URL?
    \answerNA{}
  \item Did you discuss whether and how consent was obtained from people whose data you're using/curating?
    \answerNA{}
  \item Did you discuss whether the data you are using/curating contains personally identifiable information or offensive content?
    \answerNA{}
\end{enumerate}

\item If you used crowdsourcing or conducted research with human subjects...
\begin{enumerate}
  \item Did you include the full text of instructions given to participants and screenshots, if applicable?
    \answerNA{}
  \item Did you describe any potential participant risks, with links to Institutional Review Board (IRB) approvals, if applicable?
    \answerNA{}
  \item Did you include the estimated hourly wage paid to participants and the total amount spent on participant compensation?
    \answerNA{}
\end{enumerate}
\end{enumerate}
\fi

\newpage
\appendix

\paragraph{Roadmap.} In \append{qwalk_toolbox}, we provide more details of quantum walk and give our user-friendly framework. In \append{classical}, we introduce the classical method for optimizing approximately convex functions in a self-contained way. In \append{quantum}, we prove our main result of quantum approximately convex optimization. 
\section{Basic Facts about Quantum Walk}\label{append:qwalk_toolbox}
In this section, we first define the quantum walk operators and introduce some spectral properties. Then, we show how to efficiently implement a quantum walk.  
\subsection{Definitions and spectral properties of quantum walk}
Let $P$ be the transition operator of the classical Markov chain over the space $K$ such that
\begin{align*}
  \int_K P(x,y)\d y = 1~~~\forall x\in K.
\end{align*}
We define the following states, which capture key properties of the quantum walk:
\begin{align*}
  \ket{\psi_x} := \int_{K}  \sqrt{P(x,y)}\ket{y}\d y~~~\forall x\in K.
\end{align*}
\begin{definition}[Quantum walk operators]\label{def:def_W}
  The quantum walk uses the following three operators:
  \begin{itemize}
    \item $U:=\int_K \ket{x}\ket{\psi_x}\bra{x}\bra{0}\d x$ for any $x\in K$.
    \item $\Pi:=\int_K \ket{x}\ket{\psi_x}\bra{x}\bra{\psi_x}\d x$ is the projection to the subspace $\mathrm{span}\{\ket{x}\ket{\psi_x}\}_{x\in K}$.
    \item $S:=\int_K \int_K \ket{y}\ket{x}\bra{x}\bra{y}\d x \d y$ is to swap the two quantum registers.
  \end{itemize}
  Then, the quantum walk operator $W$ is defined by:
  \begin{align*}
    W:=S(2\Pi - I).
  \end{align*}
\end{definition}

\begin{definition}[Alternative definition of quantum walk operator, \cite{wocjan2008speedup}]\label{def:W_alter}
  Define the quantum walk operator $$W':=U^\dagger S U R_A U^\dagger S U R_A,$$ where $R_{\cal A}$ denotes the reflection about the subspace ${\cal A}:=\{\ket{x}\ket{0}~|~x\in K\}$ for random walk space $K$, $S$ is the swap operator, and $U$ is the following operator:
  \begin{align*}
    U\ket{x}\ket{0}=\int_{y\in K} \sqrt{P(x,y)}\ket{x}\ket{y} \d y,
  \end{align*}
  for $P$ being the transition operator of the Markov chain.
  \end{definition}
  
\begin{fact}[Equivalence of the definitions, \cite{cch19}]\label{fac:W_alt_eigs}
    Let $W$ be defined as in Definition~\ref{def:def_W} and let $W'$ be defined as in Definition~\ref{def:W_alter}. Then, $W'$ and $W$ have the same set of eigenvalues.
\end{fact}
  
\begin{fact}[\cite{cch19}]\label{fac:P_eigs}
 Let $D$ be the discriminant operator of $P$ defined as $D(x,y):=\sqrt{P(x,y)P(y,x)}$. Then, $P$ and $D$ have the same set of eigenvalues.
\end{fact}
  
\begin{fact}[\cite{cch19}]\label{fac:W_eigs}
    Let $\{\lambda_j\}$ be the eigenvalues of $D$. Then, the eigenvalues of $W$ are
    \begin{align*}
      \left\{\pm 1, \lambda_j + \sqrt{1-\lambda_j^2}i\right\}.
    \end{align*}
\end{fact}

The following lemma shows that when the initial stationary is a warm start, then the eigenvalues whose eigenspaces have big overlap with the initial state are bounded away from 1. 
\begin{lemma}[Effective spectral gap for warm start, {\cite[Lemma C.7]{cllwz22}}]\label{lem:effect_spectral_gap}
Let $M = (\Omega, p)$ be an ergodic reversible Markov chain with a transition operator
$P$ and unique stationary state with a corresponding density $\rho$. Let $\{(\lambda_i, f_i)\}$ be the set
of eigenvalues and eigenfunctions of $P$, and $\ket{\psi_i}$ be the eigenvectors of the corresponding quantum walk operator $W$. Let $\rho_0$ be a probability density that is a warm start for $\rho$ and mixes up to TV-distance $\epsilon$ in $t$ steps of $M$. Furthermore, assume that $\rho_0$ is a $\beta$-warm start of $\rho$. 

Let $\ket{\phi_{\rho_0}}$ be the resulting state of applying the quantum walk update operator $U$ to the state $\ket{\rho_0}$:
\begin{align*}
  \ket{\phi_{\rho_0}} = \int_\Omega \sqrt{\rho_0(x)}\int_\Omega \sqrt{P(x,y)}\ket{x}\ket{y}\d x \d y.
\end{align*}
Then, we have $|\langle \phi_{\rho_0}|\psi_i\rangle| = O(\beta\sqrt{\epsilon})$ for all $i$ with $1>\lambda_i \geq 1-O(1/t)$.
\end{lemma}

Lemma~\ref{lem:effect_spectral_gap} also applies to the initial distribution with a bounded $\ell_2$-warmness:
\begin{lemma}[Effective spectral gap for $\ell_2$-warm start, {\cite{clw19}}]\label{lem:effect_spectral_gap_l2}
  Let $M = (\Omega, p)$ be an ergodic reversible Markov chain with a transition operator
  $P$ and unique stationary state with a corresponding density $\rho$. Let $\{(\lambda_i, f_i)\}$ be the set
  of eigenvalues and eigenfunctions of $P$, and $\ket{\psi_i}$ be the eigenvectors of the corresponding quantum walk operator $W$. Let $\rho_0$ be a probability density that is a warm start for $\rho$ and mixes up to TV-distance $\epsilon$ in $t$ steps of $M$. Furthermore, assume that $\|\rho / \rho_0\|=\int_\Omega \frac{\rho(x)}{\rho_0(x)}\rho(x)\d x \leq \gamma$. 
  
  Let $\ket{\phi_{\rho_0}}$ be the resulting state of applying the quantum walk update operator $U$ to the state $\ket{\rho_0}$:
  \begin{align*}
    \ket{\phi_{\rho_0}} = \int_\Omega \sqrt{\rho_0(x)}\int_\Omega \sqrt{P(x,y)}\ket{x}\ket{y}\d x \d y.
  \end{align*}
  Then, we have $|\langle \phi_{\rho_0}|\psi_i\rangle| = O(\gamma^{1/4}\epsilon^{3/4}+\sqrt{\epsilon})$ for all $i$ with $1>\lambda_i \geq 1-O(1/t)$.
  \end{lemma}

\subsection{Efficient implementation of quantum walk}

The goal of this section is to give a user-friendly quantum walk implementation cost-analysis (Theorem~\ref{thm:qwalk_cost}). 

\begin{lemma}[Approximate reflector, {\cite[Corollary 4.1]{cch19}}]\label{lem:approx_reflect}
  Let $W$ be a unitary operator with a unique leading eigenvector $\ket{\psi_0}$ with eigenvalue 1. Denote the remaining eigenvectors by $\ket{\psi_j}$ with corresponding eigenvaluese $e^{2\pi i\xi_j}$ for $j\ge 1$. For any $\Delta \in (0,1]$ and $\epsilon < 1/2$, define $a:=\log(1/\Delta)$ and $c:=\log (1/\sqrt{\epsilon})$. Let $R$ be the reflector such that $R=\alpha\ket{\psi_0}\bra{\psi_0} + (I-\ket{\psi_0}\bra{\psi_0})$.

  For any constant $\alpha\in \mathbb{C}$, there exists a quantum circuit $\wt{R}$ that uses $a\cdot c$ ancilla qubits and invokes the controlled-$W$ gate $2^{a+1} c$ times such that
  \begin{itemize}
    \item $\wt{R}\ket{\psi_0}\ket{0}^{\otimes ac}=R\ket{\psi_0}\ket{0}^{\otimes ac}$.
    \item $\big\|\wt{R}\ket{\psi_j}\ket{0}^{\otimes ac}-R\ket{\psi_j}\ket{0}^{\otimes ac}\big\|_2\leq \sqrt{\epsilon}$ for $j\geq 1$ with $\xi_j\geq \Delta$.
  \end{itemize}
\end{lemma}

\begin{lemma}[$\pi/3$-amplitude amplification, {\cite[Lemma 1]{wocjan2008speedup}}]\label{lem:amplitude_amp}
  Let $\ket{\psi}, \ket{\phi}$ be two quantum states with $|\langle \psi|\phi\rangle|\geq p$ for some $p\in (0,1]$. Let $\omega=e^{i\pi/3}$. Define $R_\psi:= \omega \ket{\psi}\bra{\psi}+(I-\ket{\psi}\bra{\psi})$, and $R_\phi:= \omega \ket{\phi}\bra{\phi}+(I-\ket{\phi}\bra{\phi})$. Then, for $m\geq 1$, there exists a sequence of unitaries:
  \begin{align*}
    V_0=I, ~~ V_{j+1} = V_jR_{\psi}V_j^\dagger R_{\phi} V_j~~~\forall j\in [m],
  \end{align*}
  such that 
  \begin{align*}
    |\langle \psi|V_m|\phi\rangle|^2\geq 1-(1-p)^{3^m}.
  \end{align*}
  Furthermore, the unitaries $R_\phi, R_\psi$ and their inverses are used at most $3^m$ times in $V_m$.  
\end{lemma}

\begin{theorem}[Quantum walk implementation cost]\label{thm:qwalk_cost}
  Let $M_0,M_1$ be two ergodic reversible Markov chains with stationary distributions $\pi_0,\pi_1$, respectively. Suppose $\pi_0$ is $\beta_0$-warm with respect to $M_1$ and mixes up to total variation distance $\epsilon$ in $t_0(\epsilon)$ steps. Similarly, suppose $\pi_1$ is $\beta_1$-warm with respect to $M_0$ and mixes in $t_1(\epsilon)$ steps. Let $\beta:=\max\{\beta_0, \beta_1\}$. Moreover, we assume that $| \langle \pi_0|\pi_1\rangle|\geq p$. 

  Given $\ket{\pi_0}$, we can obtain a state $\ket{\wt{\pi}_1}$ such that $\|\ket{\wt{\pi}_1}-\ket{\pi_1}\|_2\leq \epsilon$ using
  \begin{align*}
    O\left(\sqrt{t_0(\epsilon)+t_1(\epsilon)}\cdot p^{-1}\log(\beta/p)\log^2(1/(p\epsilon))\right)
  \end{align*}
  calls to the controlled walk operators controlled-$W_0'$, controlled-$W_{1}'$.
\end{theorem}

\begin{proof}
  By assumption, we know that $\pi_0$ mixes in $t_0'=t_0(\epsilon_1/\beta_0^2)$ steps in $M_1$ to achieve total variation distance $\epsilon_1/\beta_0^2$, where $\epsilon_1$ is a parameter to be chosen later. Similarly, $\pi_1$ mixes in $t_1'=t_1(\epsilon_1/\beta_0^2)$ steps in $M_0$ to achieve total variation distance $\epsilon_1/\beta_1^2$.

  We start from $\ket{\pi_0}$. By Lemma~\ref{lem:effect_spectral_gap},  we have $\ket{\pi_0} = \ket{\pi_{0,\mathsf{good}}}+\ket{e_0}$, where $\ket{\pi_{0,\mathsf{good}}}$ lies in the subspace spanned by the eigenvectors $\ket{\psi_j}$ of $W_1'$ with corresponding eigenvalue $\lambda_j$ of $P_1$ such that $\lambda_j=0$ or $\lambda_j\leq 1-\Omega(1/t_0')$. Let $e^{2\pi i \xi_j}$ be the eigenvalue of $\ket{\psi_j}$ of $W_1'$. By Fact~\ref{fac:W_alt_eigs} and Fact~\ref{fac:W_eigs}, we get that $\xi_j = 0$ or $\xi_j\geq \Omega({t_0'}^{-1/2})$. By Lemma~\ref{lem:effect_spectral_gap}, we also have $\|\ket{e_0}\|\leq \epsilon_1$.

  Then, by Lemma~\ref{lem:approx_reflect} with $\Delta=\Omega({t_0'}^{-1/2})$ and $\epsilon=\epsilon_1^2$, we can implement $\wt{R}_1$ such that $\|R_1\ket{\phi}-\wt{R}_1\ket{\phi}\|_2\leq 2\epsilon_1$ using $O(\sqrt{t_0'}\log(1/\epsilon_1))$ calls to controlled-$W_1'$, where $\ket{\phi}$ is any state that occurs during $\pi/3$-amplitude amplification (Lemma~\ref{lem:amplitude_amp}) for $\ket{\pi_0}$ towards $\ket{\pi_1}$. 
  
  In the same way, we can start from $\ket{\pi_1}$ and show that $\wt{R}_0$ can be implemented using  $O(\sqrt{t_1'}\log(1/\epsilon_1))$ calls to controlled-$W_0'$ such that $\|R_0\ket{\phi'}-\wt{R}_0\ket{\phi'}\|\leq 2\epsilon_1$, where $\ket{\phi'}$ is any state that occurs during $\pi/3$-amplitude amplification for $\ket{\pi_1}$ towards $\ket{\pi_0}$. 

  Suppose we can implement $R_0$ and $R_1$ perfectly. Then, we can prepare a state $\ket{\wt{\pi}_1}$ such that $|\langle \wt{\pi}_1|\pi_1\rangle|\geq 1-(1-p)^{3^m}$ using $3^m$ calls to $R_0, R_1$ and their inverses, by applying $\pi/3$-amplitude amplification (Lemma~\ref{lem:amplitude_amp}) to $\ket{\pi_i}$. Thus, by taking $m=O(p^{-1}\log(1/\epsilon_2))$ where $\epsilon_2$ is a parameter to be chosen later, we have $\|\ket{\pi_1}-\ket{\wt{\pi}_1}\|_2\leq \epsilon_2$. However, since each call to $\wt{R}_0$ or $\wt{R}_1$ causes an error of $\epsilon_1$, the total error will be
  \begin{align*}
    O(\epsilon_2 + \epsilon_1\cdot p^{-1}\log(1/\epsilon_2))= \epsilon,
  \end{align*}
  where we take $\epsilon_1:=O(p\epsilon\log^{-1}(1/\epsilon))$ and $\epsilon_2:=\epsilon_1^2$.

  Therefore, the total number of calls to controlled-$W_0'$, controlled-$W_{1}'$ is
  \begin{align*}
    O\left((\sqrt{t_0'}+\sqrt{t_1'})\cdot p^{-1}\log^2(1/p\epsilon)\right),
  \end{align*}  
  where $t_i'=t_i(\epsilon_1/\beta_i^2)=O(t_i(\epsilon)\cdot \log(\beta_i/p))$.

  The theorem is then proved.
\end{proof}

The following corollary is an immediate consequence of Theorem~\ref{thm:qwalk_cost}, and it also gives Theorem~\ref{thm:qwalk_intro}.
\begin{corollary}[Quantum walk implementation cost ($\ell_2$-warm starts)]\label{cor:qwalk_cost_l2}
  Let $M_0,M_1$ be two ergodic reversible Markov chains with stationary distributions $\pi_0,\pi_1$, respectively. Suppose $\pi_0$ mixes towards $\pi_1$ in $M_1$ up to total variation distance $\epsilon$ in $t_0(\epsilon)$ steps. Similarly, suppose $\pi_1$ mixes towards $\pi_0$ in $M_0$ in $t_1(\epsilon)$ steps. Suppose $\|\pi_0/\pi_1\|=O(1)$ and $\|\pi_1/\pi_0\|=O(1)$. Moreover, we assume that $| \langle \pi_0|\pi_1\rangle|=\Omega(1)$. 

  Given $\ket{\pi_0}$, we can obtain a state $\ket{\wt{\pi}_1}$ such that $\|\ket{\wt{\pi}_1}-\ket{\pi_1}\|_2\leq \epsilon$ using
  \begin{align*}
    O\left(\sqrt{t_0(\epsilon)+t_1(\epsilon)}\log^2(1/\epsilon)\right)
  \end{align*}
  calls to the controlled walk operators controlled-$W_0'$, controlled-$W_{1}'$.
\end{corollary}

\section{Classical Approach for Optimizing Approximately Convex Functions}\label{append:classical}
In this section, we introduce the classical approach \cite{belloni2015escaping} for the optimization of approximately convex functions as in Eq.~\eqref{eq:def_approx_convex}.

\subsection{Low level: Hit-and-Run for approximate log-concave distributions}

\begin{algorithm}[htbp]
\caption{Hit-and-Run walk}\label{alg:hit-and-run}
\begin{algorithmic}[1]
\Procedure{HitAndRun}{$\pi_0$, $\pi_g$, $\Sigma$, $m$}
\Comment{$\pi_g$ is the target distribution on ${\cal K}$ induced by a nonnegative function $g$, $\Sigma$ is a linear transformation}  
    \State ${\bf x}_0\gets $ sample from $\pi_0$
    \State Choose accuracy parameter $\epsilon_\ell$
    \For{$i\gets 1,\dots,m$}
        \State ${\bf u}\gets $ uniformly sample from the surface of ellipse given by $\Sigma$ acting on sphere
        \State $\ell(t):={\bf x}_{i-1}+t{\bf u}$, compute $[{\bf s}, {\bf t}]\gets \ell \cap {\cal K}$
        \State ${\bf x}_i\gets $ \textsc{UniSampler}($g$, $\beta$, $[{\bf s}, {\bf t}]$, $\epsilon_\ell$)
    \EndFor
    \State \Return ${\bf x}_m$
\EndProcedure
\end{algorithmic}
\end{algorithm}

The Hit-and-Run walk uses a unidimensional rejection sampler to sample a point from the distribution $\pi_g$ restricted to a line $\ell$. The following lemma shows the performance guarantee of the unidimensional sampler:

\begin{lemma}[Unidimensional rejection sampler, {\cite[Lemma 5]{belloni2015escaping}}]\label{lem:unidimensional}
    Given $\beta=O(1)$. Let $g$ be a $\beta$-log-concave function and $\ell$ be a bounded line segment on ${\cal K}$. For $\epsilon\in (0, e^{-2\beta}/2)$,  Algorithm~\ref{alg:unisampler} outputs a point ${\bf x}\in \ell$ with a distribution $\wt{\pi}_\ell$ such that
    \begin{align*}
        d_\mathrm{TV}(\wt{\pi}_\ell, \pi_g|_\ell)\leq 3e^{2\beta}\epsilon.
    \end{align*}
    Moreover, the algorithm requires $\wt{O}(1)$ evaluations of the function $g$.
\end{lemma}

\begin{algorithm}[htbp]
\caption{Unidimensional rejection sampler}\label{alg:unisampler}
\begin{algorithmic}[1]
\Procedure{InitP}{$g$, $\beta$, $\ell=[{\bf s}, {\bf t}]$}
    \While{true}
        \State ${\bf x}_1\gets \frac{3}{4}{\bf s} + \frac{1}{4}{\bf t}$, ${\bf x}_2\gets \frac{1}{2}{\bf s} + \frac{1}{2}{\bf t}$, ${\bf x}_3\gets \frac{1}{4}{\bf s} + \frac{3}{4}{\bf t}$
        \If{$|\log(g({\bf x}_1))-\log(g({\bf x}_3))|>\beta$}
            \State ${\bf t}\gets {\bf x}_3$ if $g({\bf x}_1)>g({\bf x}_3)$; ${\bf s}\gets {\bf x}_1$ otherwise
        \ElsIf{$|\log(g({\bf x}_1))-\log(g({\bf x}_2))|>\beta$}
            \State ${\bf t}\gets {\bf x}_2$ if $g({\bf x}_1)>g({\bf x}_2)$; ${\bf s}\gets {\bf x}_1$ otherwise
        \ElsIf{$|\log(g({\bf x}_2))-\log(g({\bf x}_3))|>\beta$}
            \State ${\bf t}\gets {\bf x}_3$ if $g({\bf x}_2)>g({\bf x}_3)$; ${\bf s}\gets {\bf x}_2$ otherwise
        \Else
            \State \Return ${\bf p}\gets \arg\max_{{\bf x}\in \{{\bf x}_1,{\bf x}_2,{\bf x}_3\}} g({\bf x})$
        \EndIf
    \EndWhile
\EndProcedure
\Procedure{BinSearch}{$g$,${\bf x}_l$, ${\bf x}_r$, $V_l$, $V_r$}
    \While{true}
    \State ${\bf x}_m\gets ({\bf x}_l+{\bf x}_r)/2$
    \If{$g({\bf x}_m)>V_r$}
        \State ${\bf x}_r\gets {\bf x}_m$
    \ElsIf{$g({\bf x}_m)<V_l$}
        \State ${\bf x}_l\gets {\bf x}_m$
    \Else
        \State \Return ${\bf x}_m$
    \EndIf
    \EndWhile
\EndProcedure
\Procedure{InitE}{$g$, $\beta$, $\ell=[{\bf s}, {\bf t}]$, ${\bf p}$, $\epsilon_\ell$}
\If{$g({\bf s})\geq \frac{1}{2}e^{-\beta}\epsilon_\ell g({\bf p})$}
\State ${\bf e}_0\gets {\bf s}$
\Else
\State ${\bf e}_0\gets \textsc{BinSearch}(g,{\bf s}, {\bf p}, \frac{1}{2}e^{-\beta}\epsilon_\ell g({\bf p}), \epsilon_\ell g({\bf p}))$
\EndIf

\If{$g({\bf t})\geq \frac{1}{2}e^{-\beta}\epsilon_\ell g({\bf p})$}
        \State ${\bf e}_1\gets {\bf t}$
    \Else
        \State ${\bf e}_1\gets \textsc{BinSearch}(g,{\bf p}, {\bf t}, \frac{1}{2}e^{-\beta}\epsilon_\ell g({\bf p}), \epsilon_\ell g({\bf p}))$
    \EndIf
\State \Return ${\bf e}_0, {\bf e}_1$
\EndProcedure

\Procedure{UniSampler}{$g$, $\beta$, $\ell=[{\bf s},{\bf t}]$, $\epsilon_\ell$}
    \State ${\bf p}\gets$ \textsc{InitP}($g$, $\beta$, $\ell$)
    \State ${\bf e}_0, {\bf e}_1 \gets$ \textsc{InitE}($g$, $\beta$, $\ell$, ${\bf p}$, $\epsilon_\ell$)
    \While{true}
        \State ${\bf x}\gets \mathrm{Uniform}([{\bf e}_0, {\bf e}_1])$, $r\gets \mathrm{Uniform}([0,1])$
        \If{$r\leq g({\bf x})/(e^{3\beta}g({\bf p}))$}
            \State \Return ${\bf x}$
        \EndIf
    \EndWhile
\EndProcedure    
\end{algorithmic}
\end{algorithm}

The following theorem gives the mixing time of the standard Hit-and-Run walk for an approximate log-concave distribution, where we assume that in each step we directly sample from the restricted distribution $\pi_g|_\ell$. 

\begin{theorem}[Mixing time of Hit-and-Run for approximate log-concave distribution, {\cite[Theorem 4]{belloni2015escaping}}]\label{thm:hit_and_run_mixing}
    Let $\pi_g$ be the stationary measure associated with the Hit-and-Run walk based on a $\beta/2$-approximately log-concave function $g$, and let $\sigma^{(0)}$ be an initial distribution with $\ell_2$-warmness $M := \|\sigma^{(0)}/\pi_g\|$. There is a universal constant $C$ such that for any $\gamma\in(0,1/2)$, if
    \begin{align*}
        m\geq Cn^2\frac{e^{6\beta}R^2}{r^2}\log^4\Big(\frac{e^{\beta}MnR}{r\gamma^2}\Big)\log\Big(\frac{M}{\gamma}\Big),
    \end{align*}
    then $m$ steps of the Hit-and-Run random walk based on $g$ yield
    \begin{align*}
        d_\mathrm{TV}(\sigma^{(m)},\pi_g)\leq \gamma.
    \end{align*}
\end{theorem}

The next theorem shows the closeness between the output distribution of Algorithm~\ref{alg:hit-and-run} and the target distribution $\pi_g$. Due to the unidimensional rejection sampler (Algorithm~\ref{alg:unisampler}), the stationary distribution of Algorithm~\ref{alg:hit-and-run} may not be exactly $\pi_g$. Nevertheless, we can still show that it will not deviate a lot.

\begin{theorem}[The effect of the rejection sampler, {\cite[Theorem 5]{belloni2015escaping}}]\label{thm:hit_and_run_mixing_real}
    Let $\pi_g$, $\sigma^{(0)}$ be defined as in Theorem~\ref{thm:hit_and_run_mixing}. Let $\hat{\sigma}^{(m)}$ denote the output distribution of Algorithm~\ref{alg:hit-and-run} with initial distribution $\hat{\sigma}^{(0)}$ in $m$ steps. Let $\epsilon_\ell$ be the accuracy parameter for the unidimensional rejection sampler (Algorithm~\ref{alg:unisampler}). Then, we have
    \begin{align*}
        d_\mathrm{TV}(\hat{\sigma}^{(m)}, \sigma^{(m)})\leq m\epsilon_\ell + 2d_\mathrm{TV}(\hat{\sigma}^{(0)}, \sigma^{(0)}).
    \end{align*}
    
    In particular, for $\gamma\in (0, 1/e)$, suppose $d_\mathrm{TV}(\hat{\sigma}^{(0)}, \sigma^{(0)})\leq \gamma/8$. Let $s\in (0,1)$ be such that $H_s\leq \gamma/4$, where $H_s$ is defined to be:
    $$H_s:=\sup_{A\subset {\cal K}: \pi_g(A)\leq s}~|\pi_g(A)-\sigma^{(0)}(A)|.$$  Then, there is a constant $C'$ such that, if we take $\epsilon_\ell:=\gamma e^{-2\beta}/(12 m)$ and 
    \begin{align*}
        m\geq  C' n^2 \frac{e^{6\beta}R^2}{r^2}\log^4\Big(\frac{e^{\beta}nR}{rs}\Big)\log(1/s),
    \end{align*}
    we have
    \begin{align*}
        d_\mathrm{TV}(\hat{\sigma}^{(m)}, \pi_g)\leq \gamma.
    \end{align*}
\end{theorem}
\subsection{Mid level: rounding into isotropic position}
The following lemma rounds a $\beta$-log-concave distribution to near-isotropic position.
\begin{lemma}[Rounding $\beta$-log-concave distribution, {\cite[Lemma 9]{belloni2015escaping}}]\label{lem:rounding}
    Let $\pi$ be a $\beta$-log-concave distribution in $\R^n$. By taking $N=\Theta(n\log n)$ i.i.d. samples ${\bf x}_1,\dots,{\bf x}_n$ from $\pi$, we have
    \begin{align*}
        \frac{1}{2}\leq \sigma_{\min}\Big(\frac{1}{N}\sum_{i\in [N]}{\bf x}_i{\bf x}_i^\top\Big) \leq \sigma_{\max}\Big(\frac{1}{N}\sum_{i\in [N]}{\bf x}_i{\bf x}_i^\top\Big)\leq \frac{3}{2}
    \end{align*}
    holds with probability at least $1-n^{-O(1)}$.
\end{lemma} 

\subsection{High level: simulated annealing}
At high level, we run a simulated annealing for a series of functions:
\begin{align*}
    h_i(x):=\exp(-f(x)/T_i), ~~\text{and}~~g_i(x):=\exp(-F(x)/T_i),
\end{align*}
where $f,F$ satisfy Eq.~\eqref{eq:def_approx_convex} and $\{T_i\}_{i\in [K]}$ are parameters to be chosen later.

\begin{algorithm}[htbp]
\caption{Simulated annealing}\label{alg:sim_annealing}
\begin{algorithmic}[1]
\Procedure{SimAnnealing}{$K$, $\{T_i\}_{i\in [K]}$}
    \State $N\gets \Theta(n\log n)$\Comment{The number of strands}
    \State $X_0^j\sim \mathrm{Uniform}({\cal K})$ for $j=1,\dots,N$ 
    \State ${\cal K}_0\gets {\cal K}$, $\Sigma_0\gets I$
    \State $m\gets \wt{O}(n^3)$\Comment{Theorem~\ref{thm:hit_and_run_mixing_real}}
    \For{$i\gets 1,\dots, K$}
        \State $\Sigma_i'\gets$ the rounding linear transformation for $\{X_{i-1}^j\}_{j\in [N]}$ 
        \State $\Sigma_i\gets \Sigma_i'\circ \Sigma_{i-1}$
        \For{$j\gets 1,\dots,N$}
            \State $X_i^j\gets \textsc{HitAndRun}(X^{j}_{i-1}, \pi_{g_i}, \Sigma_i, m)$ \Comment{Algorithm~\ref{alg:hit-and-run}}
        \EndFor
    \EndFor
    \State \Return $\arg\min_{i\in[K],j\in [N]}F(X_i^j)$
\EndProcedure
\end{algorithmic}
\end{algorithm}

\begin{lemma}[The warmness of annealing distributions, {\cite[Lemma 8]{belloni2015escaping}}]\label{lem:warm_forward}
    Let $g(x) = exp(-F(x))$ be a $\beta$-log-concave function. Let $\pi_{g_i}$ be a distribution with density proportional to $g_i(x)=\exp(-F(x)/T_i)$, supported on ${\cal K}$ . Let $T_i:=T_{i-1}\left(1-\frac{1}{\sqrt{n}}\right)$. Then,
    \begin{align*}
        \|\pi_{g_i}/\pi_{g_{i+1}}\|\leq C_\gamma = 5\exp(2\beta/T_i).
    \end{align*}
\end{lemma}

\begin{theorem}[Sample Guarantee for the simulated annealing, {\cite[Theorem 6]{belloni2015escaping}}]
    Fix a target accuracy $\gamma\in (0,1/e)$ and let $g$ be an $\beta/2$-approximately log-concave function in $\R^n$. Suppose the simulated annealing algorithm (Algorithm~\ref{alg:sim_annealing}) is run for $K=\sqrt{n}\log(1/\rho)$ epochs with temperature parameters $T_i=(1-1/\sqrt{n})^i$ for $0 \leq i \leq K$. If the Hit-and-Run with the unidimensional sampling scheme (Algorithm~\ref{alg:hit-and-run}) is run for $m=\wt{O}(n^3)$ number of steps prescribed in Theorem~\ref{thm:hit_and_run_mixing_real}, the algorithm maintains that
    \begin{align*}
        d_\mathrm{TV}(\hat{\sigma}^{(m)}_i,\pi_{g_i})\leq e\gamma
    \end{align*}
    for each $i\in [K]$, where $\hat{\sigma}^{(m)}_i$
    is the distribution of the $m$-th step of Hit-and-Run. Here, $m$ depends polylogarithmically on $1/\rho$.
\end{theorem}

Then, we have the following optimization guarantee for the simulated annealing procedure:

\begin{theorem}[Optimization guarantee for the simulated annealing, {\cite[Corollary 1]{belloni2015escaping}}]\label{thm:classical_annealing_cost}
    Suppose $F$ is approximately convex and $|F-f|\leq \epsilon/n$ as in Eq.~\eqref{eq:def_approx_convex}. The simulated annealing
    method with $K=\sqrt{n}\log(n/\epsilon)$ epochs produces a random point $X$ such that
    \begin{align*}
        \E[f(X)]-\min_{{\bf x}\in {\cal K}}f({\bf x})\leq \epsilon,
    \end{align*}
    and thus,
    \begin{align*}
        \E[F(X)]-\min_{{\bf x}\in {\cal K}}F({\bf x})\leq 2\epsilon.
    \end{align*}
    Furthermore, the number of oracle queries required by the method is $\wt{O}(n^{4.5})$.
\end{theorem}
\section{Quantum Speedup for Optimizing Approximately Convex Functions}\label{append:quantum}
As we discussed in previous section, there are three levels for the optimization algorithm. The goal of this section is to prove Theorem~\ref{thm:approx-convex}, where we improve the classical query complexity $\wt{O}(n^{4.5})$ (Theorem~\ref{thm:classical_annealing_cost}) to quantum query complexity $\wt{O}(n^{3})$. The main idea is to use quantum walk algorithm (introduced in \append{qwalk_toolbox}) to speed-up the low level such that each sample can be generated with less queries.

\subsection{Quantum speedup for low-level}\label{append:qspeed_low_level}
In this section, we show how to use the quantum walk algorithm  to speedup the sampling procedure in the simulated annealing process. According to the framework (Corollary~\ref{cor:qwalk_cost_l2}), we first show that the each Markov chain's stationary distribution in the annealing process is a warm-start for its adjacent chains, and the Markov chains are slowly-varying. Then, we show how to implement the quantum walk operator for the Hit-and-Run walk. Finally, we prove the quantum speedup from $\wt{O}(n^3)$ classical query complexity to $\wt{O}(n^{1.5})$ quantum query complexity.

\paragraph{Warmness and overlap for the stationary distributions.}
We first show that $\pi_{g_{i}}$ is a warm-start for $\pi_{g_{i+1}}$, and vice versa.

By lemma~\ref{lem:warm_forward}, we know that $\|\pi_{g_i}/\pi_{g_{i+1}}\|\leq 5\exp(2\beta/T_i)$. Similarly, we can also bound $\|\pi_{g_{i+1}}/\pi_{g_i}\|$:
\begin{lemma}\label{lem:warm_backward}
    Let $g(x) = exp(-F(x))$ be a $\beta$-log-concave function. Let $\pi_{g_i}$ be a distribution with density proportional to $g_i(x)=\exp(-F(x)/T_i)$, supported on ${\cal K}$ . Let $T_i:=T_{i-1}\left(1-\frac{1}{\sqrt{n}}\right)$. Then,
    \begin{align*}
        \|\pi_{g_{i+1}}/\pi_{g_{i}}\|\leq 8\exp(2\beta/T_{i+1}).
    \end{align*}
\end{lemma}
\begin{proof}
    Define $Y(a):=\int_{{\cal K}}\exp(-F(x)a)\d x$. Then, we have
    \begin{align*}
        \|\pi_{g_{i+1}}/\pi_{g_{i}}\| = &~ \frac{\int_{\cal K} \exp(-F(x)(2/T_{i+1}-1/T_{i}))\d x \cdot \int_{\cal K} \exp(-F(x)/T_{i})\d x}{\Big(\int_{\cal K}\exp(-F(x)/T_{i+1})\d x\Big)^2}\\
        = &~ \frac{Y(2/T_{i+1}-1/T_i)Y(1/T_{i})}{Y(1/T_{i+1})^2}.
    \end{align*}
    Define $G(x,t):=g(x/t)^t$. Then, we have
    \begin{align*}
    G(\lambda x + (1-\lambda)x', \lambda t + (1-\lambda)t')
        = &~ g\Big(\frac{\lambda x + (1-\lambda)x'}{\lambda t + (1-\lambda)t'}\Big)^{\lambda t + (1-\lambda)t'}\\
        = &~ g\Big(\frac{\lambda t}{\lambda t + (1-\lambda)t'}\frac{x}{t} + \frac{(1-\lambda)t'}{\lambda t + (1-\lambda)t'}\frac{x'}{t'}\Big)^{\lambda t + (1-\lambda)t'}\\
        \geq &~ \exp(-\beta(\lambda t + (1-\lambda)t'))\cdot g\Big(\frac{x}{t}\Big)^{\lambda t}\cdot g\Big(\frac{x'}{t'}\Big)^{(1-\lambda)t'}\\
        =&~ \exp(-\beta(\lambda t + (1-\lambda)t')) \cdot G(x,t)^{\lambda} \cdot G(x',t')^{1-\lambda}\\
        = &~ \left(\exp(-\beta t) G(x,t)\right)^\lambda \cdot \left(\exp(-\beta t') G(x',t')\right)^{1-\lambda},
    \end{align*}
    where the inequality follows from $g$ is $\beta$-log-concave.

    By Pr\'ekopa–Leindler inequality (Theorem~\ref{thm:P_L_ineq}), it implies that
    \begin{align*}
        \int_{\cal K}G(x,\lambda t + (1-\lambda)t')\d x \geq \Big(\int_{\cal K}\exp(-\beta t)G(x,t)\d x\Big)^{\lambda} \cdot \Big(\int_{\cal K}\exp(-\beta t')G(x,t')\d x\Big)^{1-\lambda}.
    \end{align*}

    Note that
    \begin{align*}
        \int_{\cal K}G(x,t)\d x = \int_{\cal K}g\Big(\frac{x}{t}\Big)^t\d x = t^n \int_{\cal K}g(x)^t\d x = t^n \int_{\cal K}\exp(-F(x)t)\d x = t^n Y(t).
    \end{align*}
    Hence, for $\lambda=\frac{1}{2}$, we have
    \begin{align*}
        \Big(\frac{t+t'}{2}\Big)^{2n}Y\Big(\frac{t+t'}{2}\Big)^2 \geq \exp(-\beta (t+t')/2)\cdot t^n Y(t) \cdot t'^n Y(t'),
    \end{align*}
    which implies that
    \begin{align}\label{eq:Y_func_ineq}
        \frac{Y(t)Y(t')}{Y(\frac{t+t'}{2})^2}\leq &~ \exp\Big(\frac{\beta (t+t')}{2}\Big)\cdot \Big(\frac{(t+t')^2/4}{tt'}\Big)^n.
    \end{align}

    By taking $t=2/T_{i+1}-1/T_i$ and $t'=1/T_i$, we have
    \begin{align*}
        \|\pi_{g_{i+1}}/\pi_{g_{i}}\| \leq &~ \frac{Y(2/T_{i+1}-1/T_i)Y(1/T_{i})}{Y(1/T_{i+1})^2}\\
        \leq &~ \exp(2\beta/ T_{i+1})\cdot \Big(\frac{(1/T_{i+1})^2}{(2/T_{i+1}-1/T_i)(1/T_i)}\Big)^n\\
        = &~ \exp(2\beta/ T_{i+1}) \cdot \Big(\frac{1}{(2-(1-1/\sqrt{n}))(1-1/\sqrt{n})}\Big)^n\\
        = &~ \exp(2\beta/ T_{i+1}) \cdot \Big(1+\frac{1}{n-1}\Big)^n\\
        \leq &~ \exp(2\beta/ T_{i+1}) \cdot \exp(n/(n-1))\\
        \leq &~ 8\exp(2\beta/ T_{i+1}),
    \end{align*}
    where the third step follows from $T_{i+1}=T_i(1-\frac{1}{\sqrt{n}})$.

    The lemma is then proved.
\end{proof}

\begin{remark}\label{rmk:walk_warmness}
    Since we assume that $|F(x)-f(x)|\leq \eps/n$ in Eq.~\eqref{eq:def_approx_convex}, i.e., $\beta=\eps/n$, by Lemmas~\ref{lem:warm_forward} and \ref{lem:warm_backward}, we know that the warmness $M:=\max\{\|\pi_{g+i}/\pi_{g_{i+1}}\|, \|\pi_{g_{i+1}}/\pi_{g_i}\|\}$ can be bounded by $O(\exp(2\eps / (nT_{i+1})))$. Since we choose the final temperature $T_k=\eps/n$, we get that $M=O(1)$. Therefore, it satisfies the warmness condition in Corollary~\ref{cor:qwalk_cost_l2}.
    \end{remark}

\begin{theorem}[Pr\'ekopa–Leindler inequality, \cite{pre71,pre73}]\label{thm:P_L_ineq}
    Let $0 < \lambda < 1$ and let $f, g, h : \R^n\rightarrow [0,\infty)$ be measurable functions. Suppose that these functions satisfy
    \begin{align*}
        h(\lambda x+(1-\lambda) y)\geq f(x)^{\lambda}\cdot g(y)^{1-\lambda}~~~\forall x,y\in \R^n.
    \end{align*}
    Then, we have
    \begin{align*}
        \int_{\R^n}h(x)\d x \geq \Big(\int_{\R^n}f(x)\d x\Big)^{\lambda} \cdot \Big(\int_{\R^n}g(x)\d x\Big)^{1-\lambda}.
    \end{align*}
\end{theorem}

\begin{lemma}[Bound distribution overlap]\label{lem:bound_overlap}
    Let $g(x) = \exp(-F(x))$ be a $\beta$-log-concave function. Let $\pi_{g_i}$ be a distribution with density proportional to $g_i(x)=\exp(-F(x)/T_i)$, supported on ${\cal K}$ . Let $T_i:=T_{i-1}\left(1-\frac{1}{\sqrt{n}}\right)$. Then,
    \begin{align*}
        \langle \pi_i|\pi_{i+1}\rangle \geq \exp(-(\beta/T_{i+1}+1)/2).
    \end{align*}
\end{lemma}
\begin{proof}
    We can write the overlap as follows:
    \begin{align*}
        \langle \pi_{g_i}|\pi_{g_{i+1}}\rangle = &~ \frac{\int_{\cal K}\sqrt{g_i(x)g_{i+1}(x)}\d x}{(\int_{\cal K}g_i(x)\d x)^{1/2}\cdot (\int_{\cal K}g_{i+1}(x)\d x)^{1/2}}\\
        = &~ \frac{\int_{\cal K}\exp(-F(x)(1/T_i+1/T_{i+1})/2)\d x}{(\int_{\cal K}\exp(-F(x)/T_{i})\d x)^{1/2}\cdot (\int_{\cal K}\exp(-F(x)/T_{i+1})\d x)^{1/2}}\\
        = &~ \frac{Y((1/T_i+1/T_{i+1})/2)}{Y(1/T_i)^{1/2} Y(1/T_{i+1})^{1/2}},
    \end{align*}
    where $Y(t):=\int_{\cal K} \exp(-F(x) t)\d x$.

    By Eq.~\eqref{eq:Y_func_ineq}, we have
    \begin{align*}
        \frac{Y(1/T_i)Y(1/T_{T_{i+1}})}{Y((1/T_i+1/T_{i+1})/2)^2}\leq &~ \exp(\beta(1/T_i+1/T_{i+1})/2)\cdot \Big(\frac{(1/T_i+1/T_{i+1})^2/4}{1/(T_{i}T_{i+1})}\Big)^n\\
        = &~ \exp(\beta(2-1/\sqrt{n})/(2T_{i+1})) \cdot \Big(1+\frac{1}{4(n-\sqrt{n})}\Big)^n\\
        \leq &~ \exp(\frac{1}{4}\frac{\sqrt{n}}{\sqrt{n}-1})\cdot \exp(\beta/T_{i+1})\\
        \leq &~ \exp(\beta/T_{i+1}+1),
    \end{align*}
    where the second step follows from $T_{i+1}=T_i(1-1/\sqrt{n})$.

    Therefore,
    \begin{align*}
        \langle \pi_{g_i}|\pi_{g_{i+1}}\rangle \geq \exp(-(\beta/T_{i+1}+1)/2).
    \end{align*}

\end{proof}

\begin{remark}\label{rmk:walk_overlap}
    By taking $\beta=\eps/n$ and $T_i\geq \eps/n$ in Lemma~\ref{lem:bound_overlap}, we have for any $i\in [K-1]$, the overlap can be upper-bounded by:
    \begin{align*}
        \langle \pi_{g_i}|\pi_{g_{i+1}}\rangle \geq \exp(-(\beta/T_K+1)/2)=e^{-1}.
    \end{align*}
\end{remark}

\paragraph{Implementing the quantum walk operator.} We introduce how to implement the quantum walk update operator $U$ such that:
\begin{align*}
    U\ket{x}\ket{0}=\int_{\cal K} \sqrt{P_{x,y}}\ket{x}\ket{y}\d y,
\end{align*}
where $P$ is the stochastic transition matrix for the Hit-and-Run walk.

Given an input state $\ket{x}$. We first prepare an $n$-dimensional Gaussian state in an ancilla register:
\begin{align*}
    \ket{x}\ket{0} \longrightarrow \ket{x} \int_{\R^n} (2\pi)^{-n/4} \ket{z}\d z.
\end{align*}
Then, by normalizing $z$ and applying the linear transformation $\Sigma$ in another quantum register, we get that
\begin{align*}
    \ket{x} \int_{\R^n} (2\pi)^{-n/4} \ket{z}\ket{\frac{\Sigma z}{\|z\|}}\d z.
\end{align*}
If we un-compute the $\ket{z}$ register, we get that (ignoring the normalization factor):
\begin{align*}
    \ket{x}\int_{\Sigma \mathbb{S}^n} \ket{u}\d u.
\end{align*}
Next, we coherently compute the two end-points of $\ell\cap {\cal K}$ for $\ell(t):=x+ut$ in the ancilla registers:
\begin{align*}
    \int_{\Sigma \mathbb{S}^n} \d u \ket{x}\ket{u}\ket{0}\longrightarrow \int_{\Sigma \mathbb{S}^n} \d u \ket{x}\ket{u}\ket{s,t}
\end{align*}
We coherently simulate the unidimensional sampler (Algorithm~\ref{alg:unisampler}). More specifically, we can compute the points $p,e_0,e_1$ in ancilla registers:
\begin{align*}
    \int_{\Sigma \mathbb{S}^n} \d u \ket{x}\ket{u}\ket{s,t,p,e_0,e_1}
\end{align*}
Then, we prepare two unifrom distribution states in the next two ancilla qubits:
\begin{align*}
    \int_{\Sigma \mathbb{S}^n} \d u \ket{x}\ket{u}\ket{s,t,p,e_0,e_1}\int_{[0,1]^2} \ket{r',r}\d r' \d r
\end{align*}
And the next proposed point $y$ can be computed via $y:=e_0 + r'(e_1-e_0)$:
\begin{align*}
    \int_{\Sigma \mathbb{S}^n} \d u \ket{x}\ket{u}\ket{s,t,p,e_0,e_1}\int_{[0,1]^2} \ket{r',r}\d r' \d r \ket{y}
\end{align*}
Then, we check the condition $r\leq g(y)/(3^\beta g(p))$ by querying the evaluation oracle twice and use an ancilla qubit to indicate whether it is satisfied:
\begin{align*}
    \int_{\Sigma \mathbb{S}^n}\int_{[0,1]^2} \d u\d r' \d r \ket{x}\ket{u}\ket{s,t,p,e_0,e_1} \ket{r',r} \ket{y}\ket{b},
\end{align*}
where $b\in \{0,1\}$. Then, we post-select\footnote{We can measure the last qubit. If the measurement outcome is 0, we reinitialize the $r,r'$ registers and re-preapre $\ket{y}$ and $\ket{b}$. We repeat this process until we measure $b=1$.} the last qubit for $b=1$. By un-computing the registers for $u,s,t,p,e_0,e_1,r,r'$, we get the desired state:
\begin{align*}
    \ket{x}\int_{\cal K}\sqrt{P_{x,y}}\ket{y}.
\end{align*}
By Lemma~\ref{lem:unidimensional}, we get that this procedure (including the post-selection cost) takes $O(1)$ oracle queries with high probability. Therefore, we get the following lemma:
\begin{lemma}[Implementation cost of the quantum walk update operator]\label{lem:implement_q_hit_and_run}
    For the Hit-and-Run walk (Algorithm~\ref{alg:hit-and-run}) with the unidimensional sampler (Algorithm~\ref{alg:unisampler}), the quantum walk update operator $U$ can be implemented by querying the evaluation oracle $O(1)$ times.
\end{lemma}

\paragraph{$\wt{O}(n^{1.5})$-query quantum algorithm.}
We have the following theorem:
\begin{theorem}[Quantum speedup for the Hit-and-Run sampler]\label{thm:q_hit_and_run}
    Let $\gamma\in (0, 1/e)$. Let $\pi_g$ be the stationary measure associated with the Hit-and-Run walk based on a $\beta/2$-approximately log-concave function $g$. Let $T_i=(1-1/\sqrt{n})^i$ for $0\leq i\leq K$ be the annealing schedule. Suppose we use the quantum walk to implement the Hit-and-Run walk (Algorithm~\ref{alg:hit-and-run}). Then, for each $0\leq i\leq K-1$, given a state $\ket{\pi_{g_i}}$, we can produce a state $|\hat{\sigma}_i^{(m)}\rangle$ such that
    \begin{align*}
        \||\pi_{g_{i+1}}\rangle-|\hat{\sigma}_i^{(m)}\rangle\|_2\leq O(\gamma),
    \end{align*}
    using $m=\wt{O}(n^{1.5})$ calls for the evaluation oracle.
\end{theorem}
\begin{proof}
    We use the quantum walk framework in Corollary~\ref{cor:qwalk_cost_l2}.

    For the warmness, by Remark~\ref{rmk:walk_warmness}, we know that in this annealing schedule,  $\|\pi_{g_i}/\pi_{g_{i+1}}\|$ and $\|\pi_{g_{i+1}}/\pi_{g_i}\|$ are upper-bounded by some constants.

    Then, by Theorem~\ref{thm:hit_and_run_mixing}, we get that the number of steps for evolving from $\pi_{g_i}$ to $\pi_{g_{i+1}}$ and from $\pi_{g_{i+1}}$ to $\pi_{g_{i}}$ is $\wt{O}(n^3)$ classically. The proof of Theorem~\ref{thm:hit_and_run_mixing_real} implies that the stationary distribution of the Hit-and-Run walk with unidimensional sampler is very close to the original Markov chain, only causing a constant blowup to the total variation distance. Thus, for $\gamma'=O(\gamma)$, we have $t_1(\gamma'),t_2(\gamma')= \wt{O}(n^3)$ in Corollary~\ref{cor:qwalk_cost_l2}.

    By Remark~\ref{rmk:walk_overlap}, we know that in this annealing schedule, the adjacent distributions have a big overlap. In particular, we have $|\langle \pi_{g_i}|\pi_{g_{i+1}}\rangle|\geq \Omega(1)$, satisfying the condition in Corollary~\ref{cor:qwalk_cost_l2}.

    Therefore, by Corollary~\ref{cor:qwalk_cost_l2}, we get that the state $|\hat{\sigma}_i^{(m)}\rangle$ satisfying $\||\pi_{g_{i+1}}\rangle-|\hat{\sigma}_i^{(m)}\rangle\|_2\leq O(\gamma)$ can be prepared using $\wt{O}(n^{1.5})$ calls to the controlled walk operators.

    By Lemma~\ref{lem:implement_q_hit_and_run}, each call to the quantum walk operator can be implemented with $O(1)$ query to the evaluation oracle. Hence, the total query complexity is $\wt{O}(n^{1.5})$.

    The theorem is then proved.
\end{proof}

\subsection{Non-destructive rounding in the mid-level}\label{append:qspeed_mid_level}

In the middle level, we need to compute the linear transformation $\Sigma_i$ that rounds the $\beta$-logconcave distribution to near-isotropic position. Moreover, we are given access to $N$ copies of the quantum states $\ket{\pi_{g_i}}$ and we will compute $\Sigma_i$ in a non-destructive way.

Classically, by Lemma~\ref{lem:rounding}, we can take
\begin{align*}
    \Sigma'_i(x^1,\dots,x^N):=\frac{1}{N}\sum_{i=1}^N x^j {i^j}^\top,
\end{align*}
where $x_i^j$ is the $j$-th independent sample from $\pi_{g_i}$. Then, the linear transformation in the $i$-th iteration is $\Sigma_i$ composite with the linear transformation in the $(i-1)$-th iteration, i.e.,
\begin{align*}
    \Sigma_i(x_1,\dots,x_n):=\Sigma_i'(x_1,\dots,x_N)\cdot \Sigma_{i-1}.
\end{align*}

In quantum, we can use a quantum circuit to simulate the classically computation for $\Sigma_i'$ and $\Sigma_i$ coherently, which computes the following superposition state:
\begin{align}\label{eq:round_state}
    \int_{\cal K}\d x^1\cdots \int_{\cal K}\d x^N \prod_{j=1}^N \sqrt{\pi_{g_i}(x^j)} \cdot \ket{x^1}\cdots \ket{x^N} \ket{\Sigma_i(x_1,\dots,x^n)}.
\end{align}
That is, the first $N$ quantum registers contain $N$ copies of the state $\ket{\pi_{g_i}}$, and the last quantum register contains the linear transformation $\Sigma_i$. If we directly measure the last register, we can get the desired matrix, but the coherence of the quantum states $\ket{\pi_{g_i}}$ are also destroyed.

To resolve this issue, we use the following theorem of Harrow and Wei:
\begin{theorem}[Non-destructive amplitude estimation, \cite{harrow2019adaptive}]\label{thm:non_destruct_mean}
    Let $P$ be an observable. Given state $\ket{\psi}$ and reflections $R_\psi=2\ket{\psi}\bra{\psi}-I$ and $R = 2P-I$, and any $\eta > 0$, there exists a quantum algorithm that outputs $\wt{a}$, an approximation to $a:=\langle \psi|P|\psi\rangle$, so that
    \begin{align*}
        |a-\wt{a}|\leq 2\pi\frac{a(1-a)}{M}+\frac{\pi^2}{M^2}.
    \end{align*}
    with probability at least $1-\eta$ and $O(\log(1/\eta)M)$ uses of $R_\psi$ and $R$. Morover the algorithm restores the state $\ket{\psi}$ with probability at least $1 - \eta$.
\end{theorem}

Then, we can create $O(\log N)$ copies of the state in Eq.~\eqref{eq:round_state}, and non-destructively estimate the mean of the last quantum register via the procedure in \cite{cch19,harrow2019adaptive}. More specifically, we start from $\wt{O}(N)$ copies of the states $\ket{\pi_{g_{i-1}}}$, and evolve them to $\ket{\pi_{g_{i}}}$. In the same time, the reflection operator $R$ can be approximately implemented by Lemma~\ref{lem:approx_reflect}. Then, the mean value can be estimated by Theorem~\ref{thm:non_destruct_mean}. Note that we can estimate all the coordinates of $\Sigma_i$ in the same time using the non-destructive mean estimation quantum circuit. And we get that the success probability of this procedure is at least $1-1/\poly(N)$. After that, the states in the first $N$ registers will be restored. Therefore, we get that:

\begin{lemma}[Non-destructive rounding]\label{lem:qrounding}
    For $i\in [K]$, the linear transformation $\Sigma_i$ at the $i$-iteration of the annealing process (Algorithm~\ref{alg:sim_annealing}) can be obtained using $\wt{O}(N)$ copies of the states $\ket{\pi_{g_{i-1}}}$, with query complexity $\wt{O}(N\cdot {\cal C})$ where ${\cal C}$ is the cost of evolving $\ket{\pi_{g_{i-1}}}$ to $\ket{\pi_{g_i}}$. Moreover, the states $\ket{\pi_{g_{i-1}}}$ will be recovered with high probability.
\end{lemma}

\subsection{Proof of Theorem~\ref{thm:approx-convex}}\label{append:q_opt_proof}

\begin{algorithm}[ht]
    \caption{Quantum speedup for approximately convex optimization.}\label{alg:q_main}
\begin{algorithmic}[1]
\Procedure{QSimAnnealing}{$K$, $\{T_i\}_{i\in [K]}$}
    \State $N\gets \wt{O}(n)$\Comment{The number of strands}
    \State Prepare $N$ (approximately) copies of $\ket{\pi_0}$, denoted as $|\wt{\pi}_0^{(1)}\rangle,\dots,|\wt{\pi}_0^{(N)}\rangle$, where $\pi_0=\mathrm{Uniform}({\cal K})$\label{ln:q_init}
    \For{$i\gets 1,\dots,K$}
        \State Use the $N$ copies of the state $\ket{\pi_{i-1}}$ to nondestructively obtain the linear transformation $\Sigma_i$. Let $|\hat{\pi}_{i-1}^{(1)}\rangle,\dots,|\hat{\pi}_{i-1}^{(N)}\rangle$ denote the post-measurements states \Comment{Lemma~\ref{lem:qrounding}}\label{ln:q_rounding}
        \State Apply quantum walk with $\Sigma_i$ to evolve the states $|\hat{\pi}_{i-1}^{(1)}\rangle,\dots,|\hat{\pi}_{i-1}^{(N)}\rangle$  to $|\wt{\pi}_{i}^{(1)}\rangle,\dots,|\wt{\pi}_{i}^{(N)}\rangle$\Comment{Theorem~\ref{thm:q_hit_and_run}}\label{ln:q_evolve}
    \EndFor
    \State $x_K^j\gets$ measure the final state $|\wt{\pi}_K^{(j)}\rangle$ for $j\in [N]$
    \State \Return $\arg\min_{j\in [N]} F(x_K^j)$
\EndProcedure
\end{algorithmic}
\end{algorithm}

\begin{proof}[Proof of Theorem~\ref{thm:approx-convex}]
    The quantum algorithm for optimizing an approximately convex function is given in Algorithm~\ref{alg:q_main}. By Theorem~\ref{thm:q_hit_and_run} and Lemma~\ref{lem:qrounding}, we know that it has the same optimization guarantee as the classical procedure (Algorithm~\ref{alg:sim_annealing}). Thus, we take $K=\sqrt{n}\log(n/\epsilon)$. And the output $x_*$ of \textsc{QSimAnnealing} procedure satisfies:
    \begin{align*}
        F(x_*) - \min_{x\in {\cal K}}F(x)\leq O(\epsilon)
    \end{align*}
    with high probability.

    Then, consider the query complexity. We have $K=\sqrt{n}\log(n/\epsilon)$ stages in the annealing process. In each iteration, the quantum walk has query complexity ${\cal C}=\wt{O}(n^{1.5})$ by Theorem~\ref{thm:q_hit_and_run}. Thus, the query cost of Line~\ref{ln:q_rounding} is $\wt{O}(N{\cal C})=\wt{O}(n^{2.5})$. Also, the query cost of Line~\ref{ln:q_evolve} is also $\wt{O}(n^{2.5})$. Therefore, the total query complexity of the annealing procedure is
    \begin{align*}
        K\cdot \wt{O}(n^{2.5})=\wt{O}(n^{3}).
    \end{align*}
\end{proof} 


\end{document}